\documentclass{article}

\usepackage[final, nonatbib]{neurips_2023}

\usepackage[utf8]{inputenc} 
\usepackage[T1]{fontenc}    
\usepackage{hyperref}       
\usepackage{url}            
\usepackage{booktabs}       
\usepackage{amsfonts}       
\usepackage{nicefrac}       
\usepackage{microtype}      
\usepackage{xcolor}         
\usepackage{multicol}
\usepackage{wrapfig}

\title{Fast Approximation of Similarity Graphs \\ with Kernel Density Estimation}

\author{%
  Peter Macgregor\\
  School of Informatics\\
  University of Edinburgh\\
  United Kingdom
   \And
   He Sun \\
   School of Informatics\\
   University of Edinburgh\\
   United Kingdom
}

\usepackage{amsthm, amsmath, amssymb}
\usepackage{graphicx, subfigure}

\usepackage{algorithm}
\usepackage{algorithmic}
\usepackage{tikzit}

\tikzstyle{basic}=[fill=white, draw=black, shape=circle]
\tikzstyle{square}=[fill=white, draw=black, shape=rectangle]
\tikzstyle{big dashed}=[fill=white, draw=black, shape=circle, minimum width=1cm, dashed]
\tikzstyle{vertical ellipse dashed}=[fill=none, draw=blue, minimum width=0.75cm, minimum height=3cm, ellipse, dashed, tikzit shape=rectangle, tikzit draw=blue, tikzit fill=white]
\tikzstyle{small vertical ellipse dashed}=[fill=none, draw=blue, shape=circle, tikzit fill=white, tikzit draw=blue, dashed, minimum width=0.75cm, minimum height=1.5cm, tikzit shape=rectangle, ellipse]
\tikzstyle{tiny vertical ellipse dashed}=[fill=none, draw=blue, shape=circle, tikzit fill=white, ellipse, dashed, minimum width=0.75cm, minimum height=1cm, tikzit shape=rectangle]
\tikzstyle{red}=[fill=red, draw=black, shape=circle]
\tikzstyle{green}=[fill={rgb,255: red,0; green,128; blue,128}, draw=black, shape=circle]
\tikzstyle{blue}=[fill=blue, draw=black, shape=circle]
\tikzstyle{huge dashed}=[fill=white, draw=black, shape=circle, dashed, minimum width=2cm]
\tikzstyle{medium}=[fill=white, draw=black, shape=circle, minimum width=1cm]
\tikzstyle{pale green}=[fill={rgb,255: red,173; green,231; blue,0}, draw=black, shape=circle, minimum width=1cm]
\tikzstyle{horizontal ellipse dashed}=[fill=white, draw=black, tikzit draw=magenta, tikzit shape=rectangle, minimum width=3cm, minimum height=0.75cm, ellipse, dashed]
\tikzstyle{minsize}=[fill=white, draw=black, shape=circle, minimum width=0.75cm]
\tikzstyle{horizontal ellipse green}=[fill={rgb,255: red,191; green,255; blue,0}, draw=black, tikzit draw={rgb,255: red,191; green,255; blue,0}, tikzit shape=rectangle, minimum width=3cm, minimum height=0.75cm, ellipse, dashed]
\tikzstyle{horizontal ellipse blue}=[fill={rgb,255: red,107; green,203; blue,255}, draw=black, tikzit draw=blue, tikzit shape=rectangle, minimum width=3cm, minimum height=0.75cm, ellipse, dashed]
\tikzstyle{smallblack}=[fill=black, draw=black, shape=circle, inner sep=0 pt, minimum size=3 pt]
\tikzstyle{smallSquare}=[fill=white, draw=black, shape=rectangle, inner sep=0 pt, minimum size=6 pt]
\tikzstyle{smallCircle}=[fill=white, draw=black, shape=circle, inner sep=0 pt, minimum size=6 pt]
\tikzstyle{big vertical ellipse dashed}=[fill=none, draw=blue, shape=circle, tikzit shape=rectangle, ellipse, dashed, minimum width=0.95cm, minimum height=3.7cm]
\tikzstyle{smallred}=[fill=red, draw=red, shape=circle, inner sep=0 pt, minimum size=3 pt]
\tikzstyle{smallblue}=[fill=blue, draw=blue, shape=circle, inner sep=0pt, minimum size=3pt]

\tikzstyle{directed}=[->]
\tikzstyle{undirected}=[-, line width=1pt]
\tikzstyle{directed red}=[draw=red, ->, line width=1pt]
\tikzstyle{directed green}=[draw={rgb,255: red,0; green,128; blue,128}, ->, line width=1pt]
\tikzstyle{directed blue}=[draw=blue, ->, line width=1pt]
\tikzstyle{directed purple}=[draw={rgb,255: red,128; green,0; blue,128}, ->, line width=1pt]
\tikzstyle{undirected red}=[-, draw=red, line width=1pt]
\tikzstyle{undirected green}=[-, draw={rgb,255: red,0; green,107; blue,61}, line width=1pt]
\tikzstyle{undirected blue}=[-, draw=blue, line width=1pt]
\tikzstyle{undirected purple}=[-, draw={rgb,255: red,128; green,0; blue,128}, line width=1pt]
\tikzstyle{undirected dashed}=[-, line width=1pt, dashed]
\tikzstyle{orange dashed}=[-, draw={rgb,255: red,255; green,128; blue,0}, dashed, line width=1.5pt]
\tikzstyle{directed dash}=[->, dashed]
\tikzstyle{blue dashed}=[-, draw=blue, dashed, line width=1pt]
\tikzstyle{green dashed}=[-, draw={rgb,255: red,0; green,162; blue,0}, dashed, line width=1pt]
\tikzstyle{blue filled}=[-, fill={blue!20}, draw=blue, line width=1pt, opacity=0.5, tikzit fill=white]
\tikzstyle{red filled}=[-, fill={red!20}, line width=1pt, draw=red, opacity=0.5, tikzit fill=white]
\tikzstyle{green filled}=[-, line width=1pt, draw={rgb,255: red,0; green,107; blue,61}, opacity=0.5, tikzit fill=white, fill={rgb,255: red,149; green,255; blue,179}]
\tikzstyle{orange filled}=[-, fill={orange!20}, draw=orange, line width=1pt, opacity=0.5, tikzit fill=white]
\tikzstyle{undirected dashed}=[-, draw=black, dashed, line width=1pt]

\theoremstyle{plain}
\newtheorem{theorem}{Theorem}

\newtheorem{lemma}{Lemma}
\newtheorem{corollary}{Corollary}
\theoremstyle{definition}
\newtheorem{definition}{Definition}

\theoremstyle{remark}

\newcommand{\vol}{\mathrm{vol}}
\newcommand{\G}{\mathsf{G}}
\newcommand{\K}{\mathsf{K}}

\newcommand{\fsgalg}{\textsc{FastSimilarityGraph}}
\newcommand{\samplealg}{\textsc{Sample}}

\newcommand{\twopartdefow}[3]
{
	\left\{
		\begin{array}{ll}
			#1 & \mbox{if } #2 \\
			#3 & \mbox{otherwise.}
		\end{array}
	\right.
}

\newcommand{\graphh}{\mathsf{H}}
\newcommand{\graphk}{\mathsf{K}}
\newcommand{\graphg}{\mathsf{G}}

\begin{document}

\maketitle

\begin{abstract}
  Constructing a similarity graph from a set $X$ of  data points in $ \mathbb{R}^d$ is   the first step of many modern clustering algorithms.
However, typical constructions of a similarity graph have high time complexity, and a quadratic space dependency with respect to $|X|$.
 We address this limitation and present a new algorithmic  framework that constructs a sparse approximation of the fully connected similarity graph while preserving its
cluster structure.
Our presented algorithm is based on the kernel density estimation problem, and is applicable for   arbitrary kernel functions.
We compare our designed algorithm with the   well-known  implementations from the \texttt{scikit-learn} library  and the \texttt{FAISS} library,  and
 find that our method significantly outperforms the implementation from both libraries on a variety of datasets.
\end{abstract}

\section{Introduction}
\label{sec:introduction}

Given a set $X=\{x_1,\ldots, x_n\} \subset \mathbb{R}^d$ of  data points and a similarity function $k:\mathbb{R}^d\times\mathbb{R}^d\rightarrow\mathbb{R}_{\geq 0}$ for any pair of data points $x_i$ and $x_j$, the objective of clustering is to partition these $n$ data points into   clusters such that similar points are in the same cluster. As a fundamental  data analysis technique, clustering has been extensively studied in different disciplines   ranging from algorithms and machine learning, social network analysis, to data science and statistics.

One prominent approach for clustering data points in Euclidean space consists of two simple steps: the first step is to construct a \emph{similarity graph} $\graphk=(V,E,w)$ from  $X$, where every vertex $v_i$ of $\graphg$ corresponds to  $x_i\in X$, and different vertices $v_i$ and $v_j$ are connected by an edge with   weight $w(v_i, v_j)$ if  their similarity $k(x_i, x_j)$ is positive, or higher than some threshold. Secondly, we   apply spectral clustering on $\graphg$, and its output naturally corresponds to some clustering on $X$~\cite{NgJW01}. Because of  its out-performance over traditional clustering algorithms like $k$-means,  this approach has become one of the most popular modern clustering algorithms.

On the other side,  different constructions of similarity graphs have significant impact on  the quality  and time complexity of spectral clustering, which is clearly acknowledged and appropriately  discussed by von Luxburg~\cite{Luxburg07}. Generally speaking,  there are two types of similarity graphs: 
\begin{itemize}
\item  the first one is the $k$-nearest neighbour graph~(\textsf{$k$-NN} graph), in which every vertex $v_i$ connects to $v_j$ if $v_j$ is among the $k$-nearest neighbours of $v_i$. A \textsf{$k$-NN} graph is sparse by construction,
but loses some of the structural information in the dataset since $k$ is usually small and the added edges are unweighted.
\item the second one is the fully connected graph, in which different vertices $v_i$ and $v_j$ are connected  with weight $w(v_i, v_j) = k(x_i, x_j)$. While a fully connected graph maintains most of the distance information about $X$, this graph is dense and storing such graphs requires \emph{quadratic} memory in $n$.
\end{itemize}
 
Taking the pros and cons of the two constructions into account, one would naturally ask the  question:
\begin{center}
\emph{Is it possible to directly construct a sparse graph that preserves the cluster structure of a fully connected similarity graph? }
\end{center}
We answer this question   affirmatively, and present a fast algorithm that
constructs an approximation of the fully connected similarity graph.
Our constructed graph  consists of only $\widetilde{O}(n)$ edges\footnote{We use   $\widetilde{O}(n)$ to represent  $O\left(n \cdot \log^c(n)\right)$ for some constant $c$.}, and  preserves the cluster structure of the fully connected similarity graph.

\subsection{Our Result}
Given any set  $X=\{x_1,\ldots, x_n\} \subset \mathbb{R}^d$  and a kernel function $k:\mathbb{R}^d\times\mathbb{R}^d\rightarrow\mathbb{R}_{\geq 0}$, a fully connected similarity graph $\mathsf{K}=(V,E,w)$ of $X$ consists of $n$ vertices, and every $v_i\in V$ corresponds to $x_i\in X$;   we set  $w(v_i,v_j)\triangleq k(x_i,x_j)$ for any different $v_i$ and $v_j$. 
We introduce an efficient algorithm that constructs a sparse graph $\graphg$ \emph{directly} from  $X$, such that   $\mathsf{K}$ and $\graphg$ share the same cluster-structure, and the graph matrices for $\mathsf{K}$ and $\graphg$ have approximately the same eigen-gap. This ensures that  spectral clustering from  $\graphg$  and $\mathsf{K}$ return approximately the same result.

The design of our algorithm is based on a novel reduction from
 the approximate construction of similarity graphs to the problem of Kernel Density Estimation~(\textsf{KDE}).
This reduction  shows that any algorithm for the \textsf{KDE}  can be employed   to construct a sparse representation of a fully connected similarity graph, while preserving the cluster-structure of the input data points.  This is summarised 
as follows:
\begin{theorem}[Informal Statement of Theorem~\ref{thm:meta}]\label{thm:informal}
Given a set of data points $X = \{x_1, \ldots, x_n\} \subset \mathbb{R}^d$ as input, there is a randomised algorithm that constructs a sparse graph $\graphg$ of $X$, such that it holds  with probability at least $9/10$ that 
\begin{enumerate}
    \item  graph $\graphg$ has $\widetilde{O}(n)$ edges, 
    \item graph $\graphg$ has the same cluster structure as the fully connected similarity graph $\graphk$ of $X$. 
\end{enumerate}
The algorithm uses an approximate \textsf{KDE} algorithm as a black-box, and has running time $\widetilde{O}(T_{\mathsf{KDE}}(n, n, \epsilon))$ for $\epsilon\leq 1/(6 \log(n))$, where $T_{\mathsf{KDE}}(n, n, \epsilon)$ is the running time of  solving the \textsf{KDE} problem for $n$ data points up to a $(1+\epsilon)$-approximation. 
\end{theorem}

This result builds a novel  connection between the \textsf{KDE} and the fast construction of similarity graphs, and further represents  a state-of-the-art algorithm for constructing similarity graphs. For instance,  when employing the fast Gauss transform~\cite{LJ} as the \textsf{KDE} solver, 
Theorem~\ref{thm:informal} shows that a sparse representation of the fully connected similarity graph with the Gaussian kernel can be constructed in $\widetilde{O}(n)$ time when  $d$  is  constant.
As such, in the case of low dimensions, spectral clustering runs as fast (up to a poly-logarithmic factor) as the time needed to read the input data points. Moreover, any improved algorithm for the \textsf{KDE} would result in a faster construction of approximate similarity graphs.


\begin{figure}[t]
    \centering
    \includegraphics[width=0.95\columnwidth]{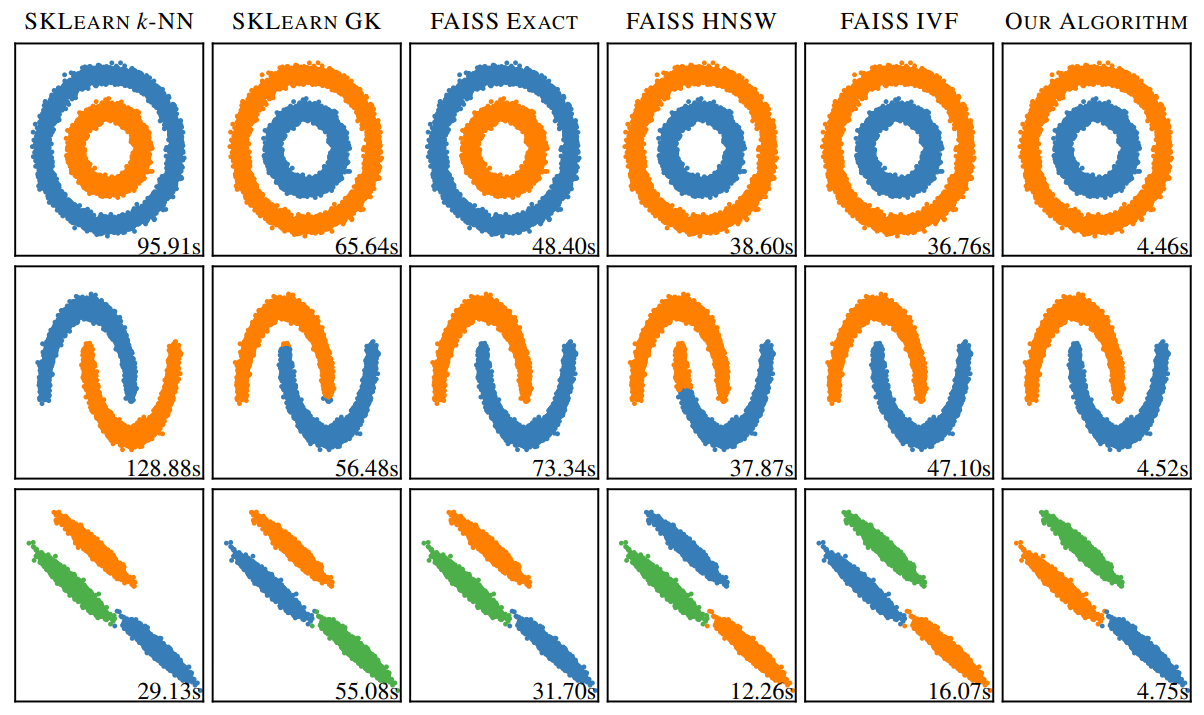}
    \caption{Output of spectral clustering with different  similarity graph constructions. 
    }
    \label{fig:dataset_comparison}
\end{figure}

To   demonstrate the significance of this work in practice, we compare  the performance of our algorithm   with  five competing algorithms from the well-known 
 \texttt{scikit-learn} library~\cite{scikit-learn} and   \texttt{FAISS} library~\cite{faissLibrary}: the algorithm that constructs a fully connected   Gaussian kernel graph, and four algorithms that construct different variants of  
  $k$-nearest neighbour graphs.
We apply spectral clustering on the  six constructed similarity graphs,
and compare the quality of the resulting clustering.
For a typical input dataset of 15,000 points in $\mathbb{R}^2$, our algorithm runs in 4.7 seconds, in comparison
with between
16.1 --  128.9 seconds for the
five competing  algorithms from \texttt{scikit-learn} and \texttt{FAISS} libraries. 
As shown in Figure~\ref{fig:dataset_comparison},
all the six algorithms return reasonable output.

We further compare the quality of the six algorithms on the BSDS image segmentation dataset~\cite{amfm_pami2011}, and  our algorithm presents clear improvement over the other five   algorithms based on the output's average Rand Index. 
In particular, due to its \emph{quadratic} memory requirement in the input size, one would need to reduce the resolution of every image down to 20,000 pixels in order to 
construct the fully connected  similarity graph with \texttt{scikit-learn} on a typical laptop. In contrast, our algorithm is able to segment the full-resolution image with over 150,000 pixels.
Our experimental result on the BSDS dataset is  showcased in 
 Figure~\ref{fig:bsds_results} and demonstrates that, in comparison with \textsc{SKLearn GK},  our algorithm identifies a more detailed pattern on the butterfly's wing. In contrast, none of the $k$-nearest neighbour based algorithms  from the two libraries is able to identify the wings of the butterfly. 

\begin{figure*}[h]
    \centering
    \subfigure[Original Image] {\includegraphics[width=0.22\textwidth]{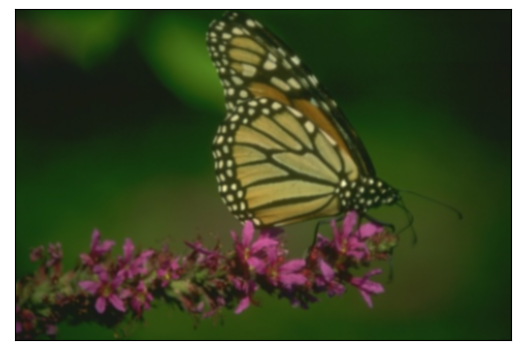}}
    \hspace{1em}
    \subfigure[\textsc{SKLearn GK}] {\includegraphics[width=0.22\textwidth]{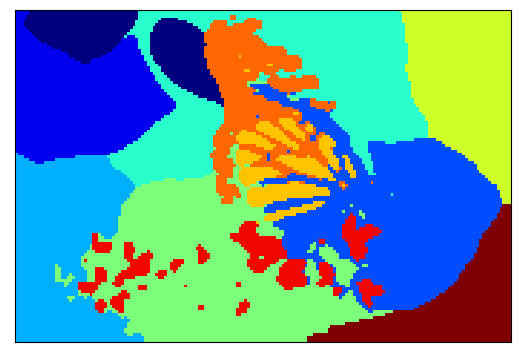}}
    \hspace{1em}
    \subfigure[\textsc{Our Algorithm}] {\includegraphics[width=0.22\textwidth]{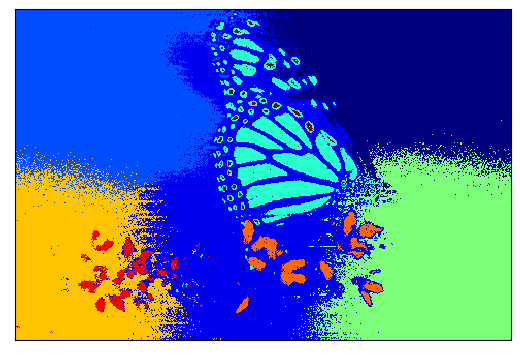}}
    \\
    \subfigure[\textsc{SKLearn $k$-NN}] {\includegraphics[width=0.22\textwidth]{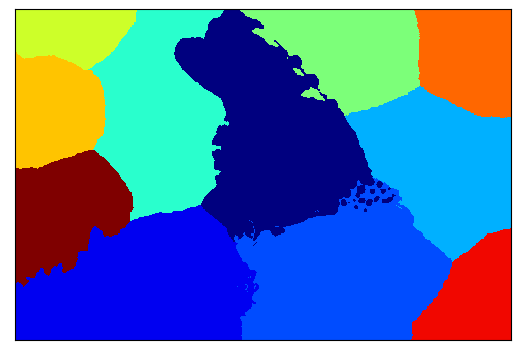}}
    \hspace{1em}
    \subfigure[\textsc{FAISS Exact}] {\includegraphics[width=0.22\textwidth]{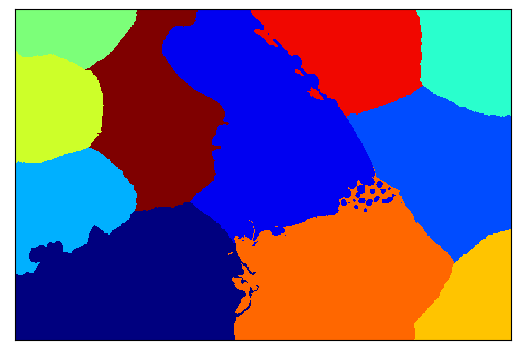}}
    \hspace{1em}
    \subfigure[\textsc{FAISS HNSW}] {\includegraphics[width=0.22\textwidth]{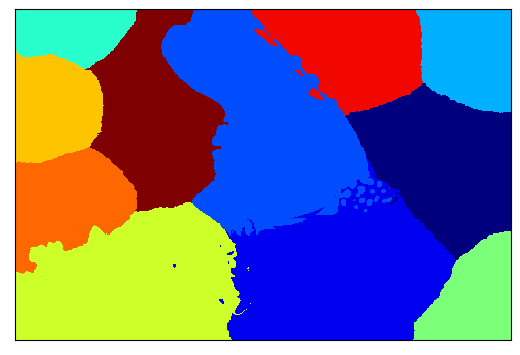}}
    \hspace{1em}
    \subfigure[\textsc{FAISS IVF}] {\includegraphics[width=0.22\textwidth]{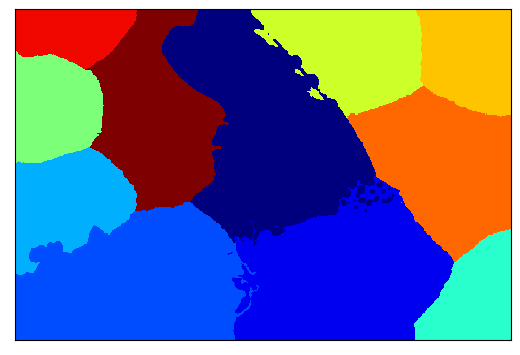}}
    

    \caption{Comparison on the performance
   of spectral clustering with   different similarity graph constructions.  Here, \textsc{SKLearn GK} is based on the fully connected similarity graph construction, and (d) -- (g) are based on different $k$-nearest neighbour graph constructions from the two libraries.  } \label{fig:bsds_results} 
\end{figure*}

\subsection{Related Work}


There is a number of work on efficient constructions of   
 $\varepsilon$-neighbourhood graphs and  \textsf{$k$-NN} graphs.  For instance, 
 Dong et al.~\cite{DongCL11} presents an  algorithm for approximate \textsf{$k$-NN} graph construction, and their algorithm is based on local search. Wang et al.~\cite{WangSWR18} presents an \textsf{LSH}-based algorithm for constructing an approximate \textsf{$k$-NN}  graph, and  employs
  several sampling and hashing techniques to reduce the  computational and
parallelisation cost. These two algorithms~\cite{DongCL11,WangSWR18} have shown to work very well in practice, but lack a theoretical guarantee on the performance. 
 
 Our work also relates to a large and  growing number of \textsf{KDE} algorithms.  
 Charikar and Siminelakis~\cite{CS17} 
  study the \textsf{KDE} problem through \textsf{LSH}, and present a class of unbiased estimators for kernel density in high dimensions for a variety of commonly used kernels.
Their work has been improved through the sketching technique~\cite{siminelakis19a}, and a revised description and analysis of the original algorithm~\cite{BackursIW19}. Charikar et al.~\cite{charikarKDE} presents a data structure for the \textsf{KDE} problem, and their result essentially matches the query time and space complexity for most studied kernels in the literature.  In addition, there are studies on designing efficient \textsf{KDE} algorithms based on 
interpolation of kernel density estimators~\cite{TurnerLR21}, and coresets~\cite{karnin19a}.

Our work further relates to  efficient constructions of  spectral sparsifiers for kernel graphs. Quanrud~\cite{Quanrud21} studies smooth kernel functions, and shows that an explicit $(1+\varepsilon)$-approximate spectral approximation of the geometric graph with $\widetilde{O}(n/\varepsilon^2)$ edges can be computed in $\widetilde{O}(n/\varepsilon^2)$ time.  
Bakshi et al.~\cite{BIKSZ} proves  that, under the strong exponential time hypothesis, constructing an $O(1)$-approximate spectral sparsifier with $O(n^{2-\delta})$ edges for the Gaussian kernel graph requires $\Omega\left(n\cdot 2^{\log (1/\tau)^{0.32}}\right)$ time, where $\delta<0.01$ is a fixed universal constant and $\tau$ is the minimum entry of the kernel matrix. 
Compared with their results, we show that, when the similarity graph with the Gaussian kernel presents a well-defined structure of clusters, an approximate construction of this similarity graph can be constructed in nearly-linear time.





\section{Preliminaries \label{sec:pre}}



Let $\graphg=(V,E, w_{\graphg})$ be an undirected  graph with weight function $w_{\graphg}:E\rightarrow\mathbb{R}_{\geq 0}$, and  $n\triangleq |V|$. 
The degree of any  vertex $v$  is defined as $\deg_{\graphg}(v)\triangleq \sum_{u\sim v} w_{\graphg}(u,v)$, where we write $u\sim v$ if $\{u,v\}\in E(\graphg)$. For any $S\subset V$, the volume of $S$ is defined by $\vol_\graphg(S)\triangleq\sum_{v\in S} \deg_\graphg(v)$, and the conductance of $S$ is defined by
\[
\phi_\graphg(S)\triangleq \frac{\partial_\graphg(S)}{\vol_\graphg(S)},
\]
where $\partial_{\graphg}(S) \triangleq \sum_{u\in S, v\not\in  S} w_{\graphg}(u,v)$.
For any $k\geq 2$,  we call subsets of vertices $A_1,\ldots, A_k$ a \emph{$k$-way partition} if $A_i\neq\emptyset$ for any $1\leq i\leq k$, $A_i \cap A_j = \emptyset$ for any $i\neq j$, and $\bigcup_{i=1}^k A_i = V$. Moreover, we define the \emph{$k$-way expansion constant} by \[
\rho_\graphg(k)  \triangleq \min_{ \textrm{partition\ } A_1,\ldots, A_k} \max_{1\leq i \leq k }\phi_\graphg(A_i).\] 
Note that a lower value of $\rho_\graphg(k)$ ensures the existence of $k$ clusters $A_1,\ldots, A_k$ of low conductance, i.e, $\graphg$ has at least $k$ clusters.


For any undirected graph $\graphg$, the adjacency matrix  $\mathbf{A}_\graphg$  of $\graphg$ is defined by $\mathbf{A}_\graphg(u,v)= w_\graphg(u,v)$ if $u\sim v$, and $\mathbf{A}_\graphg(u,v)=0$ otherwise. 
We write $\mathbf{D}_\graphg$ as  the diagonal matrix  defined by $\mathbf{D}_\graphg(v,v)= \deg_\graphg(v)$, and the normalised Laplacian of $\graphg$ is defined by 
$
\mathbf{N}_\graphg\triangleq \mathbf{I} - \mathbf{D}_\graphg^{-1/2}\mathbf{A}_\graphg\mathbf{D}_\graphg^{-1/2}$.
For any 
\textsf{PSD} matrix $\mathbf{B}\in\mathbb{R}^{n\times n}$, we write 
the eigenvalues of $\mathbf{B}$ as $\lambda_1(\mathbf{B})\leq \ldots\leq\lambda_n(\mathbf{B})$.  

It is well-known that, while computing $\rho_\graphg(k)$ exactly is \textsf{NP}-hard,  $\rho_\graphg(k)$ is closely related to   $\lambda_k$   through the higher-order Cheeger inequality~\cite{LGT14}: it holds for any $k$ that \[
\lambda_k(\mathbf{N}_\graphg)/2\leq \rho_{\graphg}(k) \leq O(k^3) \sqrt{\lambda_{k}(\mathsf{N}_\graphg)}.\]

\subsection{Fully Connected Similarity Graphs} \label{sec:prelim:fcg}

We use $X\triangleq \{x_1,\ldots x_n\}$ to represent the set of input data points, where  every  $x_i\in\mathbb{R}^d$. Given $X$ and some 
kernel function $k : \mathbb{R}^d \times \mathbb{R}^d \rightarrow \mathbb{R}_{\geq 0}$, we use $\K=(V_{\K},E_{\K},w_{\K})$ to represent the   fully connected similarity graph from $X$, which is constructed as follows:  every $v_i\in V_{\K}$ corresponds to $x_i\in X$, and any pair of different  $v_i$ and $v_j$ is connected by an edge with weight $
w_{\K}(v_i, v_j) = k(x_i, x_j)$. One of the most common kernels used for clustering is the Gaussian kernel, which is defined by 
\[
    k(x_i, x_j) = \exp\left( - \frac{\| x_i - x_j\|_2^2}{\sigma^2} \right)
\]
for some bandwidth parameter $\sigma$.
Other popular kernels include the Laplacian kernel and the exponential kernel which use $\|x_i - x_j\|_1$ and $\|x_i - x_j\|_2$ in the exponent respectively.

\subsection{Kernel Density Estimation} \label{sec:prelim:kde} 

 Our work is based on  algorithms for kernel density estimation~(\textsf{KDE}), which is defined as follows. Given a kernel function $k: \mathbb{R}^d \times \mathbb{R}^d \rightarrow \mathbb{R}_{\geq 0}$ with $n$ source points $x_1,\ldots, x_n \in \mathbb{R}^d$ and $m$ target points $y_1,\ldots, y_m \in \mathbb{R}^d$, the \textsf{KDE} problem is
to compute $g_{[1,n]}(y_1),\ldots g_{[1,n]}(y_m)$, where \begin{equation} \label{eq:gti}
g_{[a,b]}(y_i) \triangleq \sum_{j=a}^b k(y_i, x_j)
\end{equation}
for $1\leq i\leq m$.
While a direct computation of the $m$ values $g_{[1,n]}(y_1),\ldots g_{[1,n]}(y_m)$ requires $m n$ operations, 
there is substantial research to develop faster  algorithms approximating these $m$ quantities.

In this paper we are interested in the algorithms that approximately compute $g_{[1,n]}(y_i)$   
for all $1 \leq i \leq m$ up to a $(1 \pm \epsilon)$-multiplicative error, and use   $T_{\mathsf{KDE}}(m, n, \epsilon)$ to denote the asymptotic complexity of such a \textsf{KDE} algorithm. We also require that $T_{\mathsf{KDE}}(m, n, \epsilon)$ is superadditive in $m$ and $n$; that is, for $m = m_1 + m_2$ and $n = n_1 + n_2$, we have
\begin{equation}\nonumber
    T_{\mathsf{KDE}}(m_1, n_1, \epsilon) + T_{\mathsf{KDE}}(m_2, n_2, \epsilon) \leq T_{\mathsf{KDE}}(m, n, \epsilon);
\end{equation}
it is known that  such property holds for many \textsf{KDE} algorithms~(e.g.,~\cite{AlmanCS20,charikarKDE,LJ}).

\section{Cluster-Preserving Sparsifiers\label{sec:cps}}
A graph sparsifier is a sparse representation of an  input graph that inherits certain properties of the original dense graph.
The efficient construction of sparsifiers plays an important  role in designing a number of nearly-linear time graph algorithms. However, most algorithms for constructing sparsifiers rely on the recursive decomposition of an input graph~\cite{SpielmanT11}, sampling with respect to effective resistances~\cite{LS18,SS11}, or fast \textsf{SDP} solvers~\cite{LS17}; all of these  need the explicit representation of an input graph,  requiring  $\Omega(n^2)$ time and space complexity  for a fully connected graph.

 Sun and Zanetti~\cite{SZ19} study a   variant of graph sparsifiers that mainly preserve the cluster structure of an input graph, and introduce the notion of \emph{cluster-preserving sparsifier} defined as follows:
 \begin{definition}[Cluster-preserving sparsifier]\label{def:cps}
Let  $\graphk = (V, E, w_{\graphk})$  be any graph, and   $\{A_i\}_{i=1}^k$ the $k$-way partition of $\graphk$   corresponding to $\rho_\graphk(k)$. We call 
a re-weighted subgraph $\graphg=(V, F\subset E, w_{\graphg})$   a cluster-preserving sparsifier of $\graphk$ if 
 $\phi_\graphg(A_i) = O(k \cdot \phi_\graphk(A_i))$ for $1 \leq i \leq k$, 
    and $\lambda_{k+1}(\mathbf{N}_\graphg) = \Omega(\lambda_{k+1}(\mathbf{N}_\graphk))$.
\end{definition}
Notice that graph $\graphk$ has exactly $k$ clusters if (i) $\graphk$ has $k$ disjoint subsets $A_1,\ldots, A_k$ of low conductance, and (ii) any $(k+1)$-way partition of  $\graphk$ would include some   $A\subset V$ of high conductance, which would be implied by  a lower bound on $\lambda_{k+1}(\mathbf{N}_\graphk)$ due to 
  the higher-order Cheeger inequality. Together with the well-known eigen-gap heuristic~\cite{LGT14,Luxburg07} and theoretical analysis   on spectral clustering~\cite{MS22,PSZ}, the two properties in Definition~\ref{def:cps} ensures that spectral clustering returns approximately the same output from   $\graphk$ and $\graphh$.\footnote{The most interesting regime for this definition is  $k = \widetilde{O}(1)$ and $\lambda_{k+1}(\mathbf{N}_\graphk)=\Omega(1)$,   and we assume this in the rest of the paper.}

Now we present the algorithm in \cite{SZ19} for constructing a cluster-preserving sparsifier, and we call it the 
\textsf{SZ}  algorithm for simplicity. Given any input graph $\graphk=(V,E,w_{\graphk})$ with weight function $w_{\graphk}$, the algorithm computes  
\[
p_u(v)  \triangleq \min\left\{ C\cdot \frac{\log n}{\lambda_{k+1}}\cdot \frac{w_{\graphk}(u,v)}{\deg_{\graphk}(u)},1 \right\}, \quad\mbox{and}\quad  
p_v(u)  \triangleq \min\left\{ C\cdot \frac{\log n}{\lambda_{k+1}}\cdot \frac{w_{\graphk}(v,u)}{\deg_{\graphk}(v)},1 \right\},
\]
for every edge $e=\{u,v\}$, where $C\in\mathbb{R}^+$ is some constant.
Then, the algorithm samples every  edge $e=\{u,v\}$ with probability 
\[
p_e \triangleq p_u(v) + p_v(u) - p_u(v)\cdot p_v(u),\]
and sets the weight of every sampled $e=\{u,v\}$ in $\graphg$ as 
$
w_\graphg(u,v)\triangleq w_\graphk(u,v)/p_e$.
By setting  $F$ as the set  of the sampled edges, the algorithm returns $\graphg=(V, F, w_\graphg)$ as output. It is shown in \cite{SZ19} that, with high probability, the constructed $\graphg$ has $\widetilde{O}(n)$ edges and is a   cluster-preserving sparsifier of $\graphk$.
 
We remark that a cluster-preserving sparsifier is a much weaker notion than a spectral sparsifier, which approximately preserves all the cut values and the eigenvalues of the graph Laplacian matrices. On the other side,  while a cluster-preserving sparsifier is sufficient for the task of graph clustering, the \textsf{SZ} algorithm runs in  $\Omega(n^2)$ time for a fully connected input graph,  since it's based on the computation of the vertex degrees as well as the sampling probabilities $p_u(v)$ for every 
pair of vertices $u$ and $v$.

\section{Algorithm} \label{sec:algorithm}

 This section presents our algorithm that directly constructs an approximation of a fully connected similarity graph from $X\subseteq \mathbb{R}^d$ with $|X|=n$. As the main theoretical contribution, we demonstrate that neither the quadratic space complexity for directly constructing a fully connected similarity graph nor the quadratic time complexity of the \textsf{SZ} 
algorithm is necessary   when approximately constructing a fully connected similarity graph for the purpose of clustering. The performance of our  algorithm is  as  follows:
\begin{theorem}[Main Result]\label{thm:meta}
Given a set of data points $X = \{x_1, \ldots, x_n\} \subset \mathbb{R}^d$ as input, there is a randomised algorithm that constructs a sparse graph $\graphg$ of $X$, such that it holds  with probability at least $9/10$ that  
\begin{enumerate}
    \item  graph $\graphg$ has $\widetilde{O}(n)$ edges, 
    \item graph $\graphg$ has the same cluster structure as the fully connected similarity graph $\graphk$ of $X$; that is, if $\graphk$ has $k$ well-defined clusters, then it holds that     
    $\rho_{\graphg}(k)=O(k\cdot \rho_{\graphk}(k))$ and $\lambda_{k+1}(\mathbf{N}_{\graphg})=\Omega(\lambda_{k+1}(\mathbf{N}_{\graphk}))$.
\end{enumerate}
The algorithm uses an approximate \textsf{KDE} algorithm as a black-box, and has running time $\widetilde{O}(T_{\mathsf{KDE}}(n, n, \epsilon))$ for $\epsilon\leq 1/(6 \log(n))$. 
\end{theorem}


\subsection{Algorithm Description} At a very  high level, our designed algorithm applies a \textsf{KDE} algorithm as a black-box, and constructs a cluster-preserving sparsifier by simultaneous sampling of the edges from a \emph{non-explicitly constructed} fully connected graph. To explain our technique, we   first claim that, for an arbitrary   $x_i$,  a random $x_j$ can be sampled with probability $k(x_i, x_j) / \deg_\graphk(v_i)$ through $O(\log n)$  queries to a \textsf{KDE} algorithm.  To see this, notice that we can apply a KDE algorithm   to compute the probability that the sampled neighbour
is in some set $X_1 \subset X$, i.e., 
\[
    \mathbb{P}\left[z \in X_1\right] = \sum_{x_j \in X_1} \frac{k(x_i, x_j)}{\deg_\graphk(v_i)} = \frac{g_{X_1}(x_i)}{g_{X}(x_i)},
\]
where   we use $g_X(y_i)$ to denote that the \textsf{KDE} is taken with respect to the set of source points $X$. Based on this, we recursively split the set of possible neighbours in half and choose between the two subsets with the correct probability.
The sampling procedure is summarised as follows, and is illustrated in Figure~\ref{fig:algrecursion}. We remark that the method of sampling a random neighbour of a vertex in $\graphk$ through \textsf{KDE} and binary search also appears in Backurs et al.~\cite{BIKSZ}
\begin{enumerate}
    \item Set the feasible neighbours to be $X = \{x_1, \ldots, x_n\}$.
    \item While $|X| > 1$:
    \begin{itemize}
        \item Split $X$ into $X_1$ and $X_2$ with $|X_1|=\lfloor |X|/2 \rfloor$ and $ |X_2| = \lceil |X|/2 \rceil$.
        \item Compute $g_{X}(x_i)$ and $g_{X_1}(x_i)$;  set $X\gets X_1$ with probability $g_{X_1}(x_i) / g_{X}(x_i)$, and $X\gets X_2$ with probability $1-g_{X_1}(x_i) / g_{X}(x_i)$.
    \end{itemize}
    \item Return the remaining element in $X$ as the sampled neighbour.
\end{enumerate}

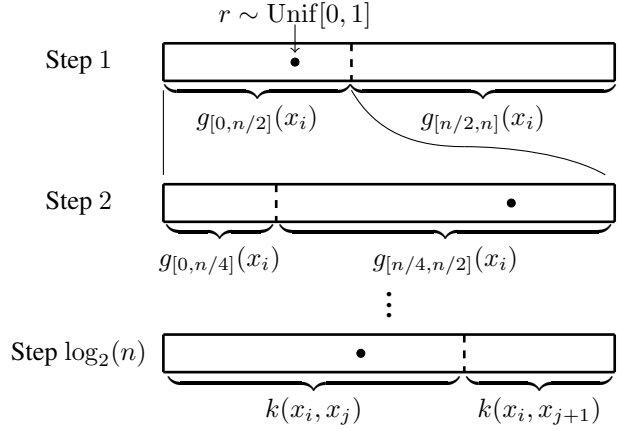
\begin{figure}
    \centering
    \begin{tikzpicture}
	\begin{pgfonlayer}{nodelayer}
		\node [style=none] (0) at (-6, 3.75) {};
		\node [style=none] (1) at (6, 3.75) {};
		\node [style=none] (2) at (6, 2.75) {};
		\node [style=none] (3) at (-6, 2.75) {};
		\node [style=none] (4) at (-8.25, 3.25) {Step $1$};
		\node [style=smallblack] (5) at (-2.5, 3.25) {};
		\node [style=none] (6) at (-2.5, 4.5) {$r \sim \mathrm{Unif}[0, 1]$};
		\node [style=none] (7) at (-1, 3.75) {};
		\node [style=none] (8) at (-1, 2.75) {};
		\node [style=none] (9) at (-3.5, 2) {$\underbrace{\hspace{7em}}_{\displaystyle g_{[0, n/2]}(x_i)}$};
		\node [style=none] (10) at (2.5, 2) {$\underbrace{\hspace{10em}}_{\displaystyle g_{[n/2,n]}(x_i)}$};
		\node [style=none] (11) at (-6, 0) {};
		\node [style=none] (12) at (6, 0) {};
		\node [style=none] (13) at (6, -1) {};
		\node [style=none] (14) at (-6, -1) {};
		\node [style=none] (15) at (-8.25, -0.5) {Step $2$};
		\node [style=smallblack] (16) at (3.25, -0.5) {};
		\node [style=none] (17) at (-6, 2.5) {};
		\node [style=none] (18) at (-1, 2.5) {};
		\node [style=none] (19) at (-6, 0.25) {};
		\node [style=none] (20) at (5.75, 0.25) {};
		\node [style=none] (21) at (-2.5, 4.25) {};
		\node [style=none] (22) at (-2.5, 3.5) {};
		\node [style=none] (23) at (-3, 0) {};
		\node [style=none] (24) at (-3, -1) {};
		\node [style=none] (25) at (-4.5, -1.75) {$\underbrace{\hspace{4em}}_{\displaystyle g_{[0, n/4]}(x_i)}$};
		\node [style=none] (26) at (1.5, -1.75) {$\underbrace{\hspace{12.5em}}_{\displaystyle g_{[n/4,n/2]}(x_i)}$};
		\node [style=none] (27) at (0, -3) {\textbf{\vdots}};
		\node [style=none] (28) at (-6, -4) {};
		\node [style=none] (29) at (6, -4) {};
		\node [style=none] (30) at (6, -5) {};
		\node [style=none] (31) at (-6, -5) {};
		\node [style=none] (32) at (-8.25, -4.5) {Step $\log_2(n)$};
		\node [style=smallblack] (33) at (-0.75, -4.5) {};
		\node [style=none] (34) at (2, -4) {};
		\node [style=none] (35) at (2, -5) {};
		\node [style=none] (36) at (-2, -5.75) {$\underbrace{\hspace{11em}}_{\displaystyle k(x_i, x_{j})}$};
		\node [style=none] (37) at (4, -5.75) {$\underbrace{\hspace{5.5em}}_{\displaystyle k(x_i, x_{j+1})}$};
		\node [style=none] (38) at (2, 1) {};
	\end{pgfonlayer}
	\begin{pgfonlayer}{edgelayer}
		\draw [style=undirected] (0.center) to (1.center);
		\draw [style=undirected] (1.center) to (2.center);
		\draw [style=undirected] (2.center) to (3.center);
		\draw [style=undirected] (3.center) to (0.center);
		\draw [style=undirected dashed] (7.center) to (8.center);
		\draw [style=undirected] (11.center) to (12.center);
		\draw [style=undirected] (12.center) to (13.center);
		\draw [style=undirected] (13.center) to (14.center);
		\draw [style=undirected] (14.center) to (11.center);
		\draw (17.center) to (19.center);
		\draw [style=directed] (21.center) to (22.center);
		\draw [style=undirected dashed] (23.center) to (24.center);
		\draw [style=undirected] (28.center) to (29.center);
		\draw [style=undirected] (29.center) to (30.center);
		\draw [style=undirected] (30.center) to (31.center);
		\draw [style=undirected] (31.center) to (28.center);
		\draw [style=undirected dashed] (34.center) to (35.center);
		\draw [in=180, out=-60] (18.center) to (38.center);
		\draw [in=150, out=0, looseness=0.75] (38.center) to (20.center);
	\end{pgfonlayer}
\end{tikzpicture}
    \caption{The procedure of sampling a neighbour $v_j$ of $v_i$ with probability $k(x_i, x_j) / \deg_\graphk(v_i)$.
    Our algorithm performs a binary search to find the sampled neighbour.
    At each step, the value of two kernel density estimates are used to determine where the sample lies.
    Notice that the algorithm doesn't  compute any edge weights directly until the last step.
    \label{fig:algrecursion}}
\end{figure}


Next we  generalise this idea and show that,  instead of   sampling a neighbour of one vertex at a time, a \textsf{KDE} algorithm allows us to  sample neighbours of every vertex in the graph ``simultaneously''.  Our designed sampling procedure is formalised in 
Algorithm~\ref{alg:fsgsample}.


\begin{algorithm}[h]
    \caption{\samplealg} 
    \begin{multicols}{2}

    \label{alg:fsgsample}
\begin{algorithmic}[1]
    \STATE {\bfseries Input:}
     set $S$ of  $\{y_i\}$ \\
    \hspace{2.7em} set $X$ of $\{x_i\}$
    \STATE {\bfseries Output:}\\
    \hspace{1em}  $E = \{(y_i, x_j)~\mbox{for some $i$ and $j$}\}$
    \IF {$\left| X \right| = 1$}
        \STATE {\bfseries return} $S \times X$
    \ELSE
        \STATE $X_1 = \{x_j : j < \left| X \right| / 2\}$
        \STATE $X_2 = \{x_j : j \geq \left| X \right| / 2\}$
        \STATE Compute $g_{X_1}(y_i)$ for all $i$ with  a \textsf{KDE} algorithm
        \STATE Compute $g_{X_2}(y_i)$ for all $i$ with a \textsf{KDE} algorithm
        \STATE $S_1 = S_2 = \emptyset$
        \FOR{$y_i \in S$}
            \STATE $r \sim \mathrm{Unif}[0, 1]$
            \IF {$r \leq g_{X_1}(y_i) / (g_{X_1}(y_i) + g_{X_2}(y_i))$}
                \STATE $S_1 = S_1 \cup \{y_i\}$
            \ELSE
                \STATE $S_2 = S_2 \cup \{y_i\}$
            \ENDIF
        \ENDFOR
        \STATE {\bfseries return} $\samplealg(S_1, X_1)\cup \samplealg(S_2, X_2)$
    \ENDIF
\end{algorithmic}
\end{multicols}
\end{algorithm}

Finally, to construct a cluster-preserving sparsifier,  we apply Algorithm~\ref{alg:fsgsample} to  sample $O(\log n)$ neighbours for every vertex $v_i$, 
and set the weight of every sampled edge $v_i\sim v_j$ as
\begin{equation}\label{eq:newweight}
    w_\graphg(v_i, v_j) = \frac{k(x_i, x_j)}{\widehat{p}(i, j)},
\end{equation}
where $\widehat{p}(i, j)\triangleq  \widehat{p}_i(j) + \widehat{p}_j(i) - \widehat{p}_i(j) \cdot \widehat{p}_j(i)$ is an estimate  of the sampling probability of   edge $v_i\sim v_j$, and 
\[
\widehat{p}_i(j) \triangleq \min\left\{6C\cdot\log n \cdot \frac{k(x_i, x_j)}{ g_{[1, n]}(x_i)}, 1\right\}
\]
for some constant $C\in\mathbb{R}^+$;  see Algorithm~\ref{alg:fsg} for the formal description.

\begin{algorithm}[t]
\begin{multicols}{2}

   \caption{\fsgalg}
   \label{alg:fsg}
\begin{algorithmic}[1]
   \STATE {\bfseries Input:} data point set $X = \{x_1, \ldots, x_n\}$
   \STATE {\bfseries Output:} similarity graph $\graphg$
   \STATE $E = \emptyset$, $L = 6 C\cdot \log (n)/\lambda_{k+1}$
   \FOR {$\ell \in [1, L]$}
        \STATE $E = E \cup \samplealg(X, X)$
   \ENDFOR
   \STATE Compute $g_{[1,n]}(x_i)$ for each $i$ with a  \textsf{KDE} algorithm
   \FOR {$(v_i, v_j) \in E$}
    \STATE $\widehat{p}_i(j) = \min\left\{L \cdot k(x_i, x_j) / g_{[1,n]}(x_i), 1\right\}$
    \STATE $\widehat{p}_j(i) = \min\left\{L \cdot k(x_i, x_j) / g_{[1,n]}(x_j), 1\right\}$
    \STATE $\widehat{p}(i, j) = \widehat{p}_i(j) + \widehat{p}_j(i) - \widehat{p}_i(j) \cdot \widehat{p}_j(i)$
    \STATE Set $w_\graphg(v_i, v_j) = k(x_i, x_j) / \widehat{p}(i, j)$
   \ENDFOR
   \STATE {\bfseries return} graph $\G = (X, E, w_\graphg)$
\end{algorithmic}
\end{multicols}
\end{algorithm}

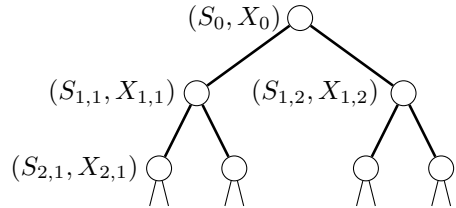
\begin{wrapfigure}{r}{0.45\textwidth}
    \centering
    \begin{tikzpicture}
	\begin{pgfonlayer}{nodelayer}
		\node [style=basic] (0) at (0, 3) {};
		\node [style=basic] (1) at (-2.75, 1) {};
		\node [style=basic] (2) at (2.75, 1) {};
		\node [style=basic] (3) at (-3.75, -1) {};
		\node [style=basic] (4) at (-1.75, -1) {};
		\node [style=basic] (5) at (1.75, -1) {};
		\node [style=basic] (6) at (3.75, -1) {};
		\node [style=none] (18) at (-4, -2) {};
		\node [style=none] (19) at (-1.5, -2) {};
		\node [style=none] (20) at (1.5, -2) {};
		\node [style=none] (21) at (4, -2) {};
		\node [style=none] (22) at (-3.5, -2) {};
		\node [style=none] (23) at (-2, -2) {};
		\node [style=none] (24) at (2, -2) {};
		\node [style=none] (25) at (3.5, -2) {};
		\node [style=none] (30) at (-1.75, 3) {$(S_0, X_0)$};
		\node [style=none] (31) at (-5, 1) {$(S_{1,1}, X_{1, 1})$};
		\node [style=none] (32) at (0.4, 1) {$(S_{1,2}, X_{1,2})$};
		\node [style=none] (33) at (-6, -1) {$(S_{2,1}, X_{2, 1})$};
	\end{pgfonlayer}
	\begin{pgfonlayer}{edgelayer}
		\draw [style=undirected] (0) to (1);
		\draw [style=undirected] (0) to (2);
		\draw [style=undirected] (1) to (3);
		\draw [style=undirected] (1) to (4);
		\draw [style=undirected] (2) to (5);
		\draw [style=undirected] (2) to (6);
		\draw (3) to (18.center);
		\draw (4) to (19.center);
		\draw (5) to (20.center);
		\draw (6) to (21.center);
		\draw (6) to (25.center);
		\draw (5) to (24.center);
		\draw (4) to (23.center);
		\draw (3) to (22.center);
	\end{pgfonlayer}
\end{tikzpicture}
    \caption{\small{The recursion tree for Algorithm~\ref{alg:fsgsample}.
    \label{fig:algtree}}}
\end{wrapfigure}

\subsection{Algorithm Analysis} Now we analyse Algorithm~\ref{alg:fsg}, and sketch the proof of Theorem~\ref{thm:meta}. 
We first analyse the running time of Algorithm~\ref{alg:fsg}. Since it involves   
$O(\log n)$ executions of Algorithm~\ref{alg:fsgsample} in total, it is sufficient  to examine the running time of Algorithm~\ref{alg:fsgsample}.

We visualise the recursion of   Algorithm~\ref{alg:fsgsample}  with respect to $S$ and $X$ in Figure~\ref{fig:algtree}. Notice that,
although the number of nodes doubles at each level of the recursion tree, the total number of samples $S$ and data points $X$ remain constant for each level of the tree:   it holds for any $i$ that
$
\bigcup_{j=1}^{2^i} S_{i,j}= S_0
$
and 
$
\bigcup_{j=1}^{2^i} X_{i,j}= X_0$. 
Since the running time of the \textsf{KDE} algorithm is superadditive, the combined running time of all nodes at level $i$ of the tree is
\begin{align*}
    T_{i} & = \sum_{j = 1}^{2^i} T_{\textsf{KDE}}(|S_{i, j}|, |X_{i, j}|, \epsilon) \leq T_{\textsf{KDE}}\left(\sum_{j = 1}^{2^i} |S_{i, j}|, \sum_{j = 1}^{2^i} |X_{i, j}|, \epsilon \right) = T_{\textsf{KDE}}(n, n, \epsilon).
\end{align*}
Hence, the total running time of Algorithm~\ref{alg:fsg} is $\widetilde{O}(T_{\textsf{KDE}}(n, n, \epsilon))$ as the tree has $\log_2(n)$ levels. Moreover,   when applying  the \emph{Fast Gauss Transform}~(\textsf{FGT})~\cite{LJ} as the \textsf{KDE} algorithm, the total running time of  Algorithm~\ref{alg:fsgsample} is $\widetilde{O}(n)$  when $d = O(1)$ and $\tau = \Omega(1 / \mathrm{poly}(n))$.


Finally, we prove that the output of Algorithm~\ref{alg:fsg} is a cluster preserving sparsifier of a fully connected similarity graph, and has $\widetilde{O}(n)$ edges.
Notice that, comparing with the sampling probabilities $p_u(v)$ and $p_v(u)$ used in the \textsf{SZ} algorithm, Algorithm~\ref{alg:fsg} samples each edge through $O(\log n)$ recursive executions of a \textsf{KDE}
 algorithm, each of which introduces a $(1+\epsilon)$-multiplicative error.
 We carefully examine these $(1+\epsilon)$-multiplicative errors and prove that the actual sampling probability $\widetilde{p}(i,j)$ used in Algorithm~\ref{alg:fsg} is an $O(1)$-approximation of $p_e$ for every $e=\{v_i, v_j\}$. As such the output of Algorithm~\ref{alg:fsg} is a cluster-preserving sparsifier of a fully connected similarity graph, and satisfies the two properties of Theorem~\ref{thm:meta}.  The total number of edges in $\graphg$ follows from the sampling scheme.
 We refer the reader to the supplementary material for the complete proof of Theorem~\ref{thm:meta}.






\section{Experiments} \label{sec:experiments}

In this section, we empirically evaluate the performance of spectral clustering with our new algorithm for constructing similarity graphs.
We compare our algorithm with the
algorithms
for similarity graph construction 
provided by the \texttt{scikit-learn} library~\cite{scikit-learn}
and the \texttt{FAISS} library~\cite{faissLibrary}
which are commonly used machine learning libraries for Python.
In the remainder of this section, we compare the following six spectral clustering methods. 
\begin{enumerate}\itemsep 0.5pt
    \item \textsc{SKLearn GK}: spectral clustering with the fully connected Gaussian kernel similarity graph constructed with the \texttt{scikit-learn} library.
    \item \textsc{SKLearn  $k$-NN}: spectral clustering with the $k$-nearest neighbour similarity graph constructed with the \texttt{scikit-learn} library.
    \item \textsc{FAISS Exact}: spectral clustering with the exact $k$-nearest neighbour graph constructed with the \texttt{FAISS} library.
    \item \textsc{FAISS HNSW}: spectral clustering with the approximate $k$-nearest neighbour graph constructed with the ``Hierarchical Navigable Small World'' method~\cite{faissHNSW}.
    \item \textsc{FAISS IVF}: spectral clustering with the approximate $k$-nearest neighbour graph constructed with the ``Invertex File Index'' method~\cite{faissIVF}.
    \item \textsc{Our Algorithm}: spectral clustering with the Gaussian kernel similarity graph constructed by Algorithm~\ref{alg:fsg}.
\end{enumerate}
 
We implement Algorithm~\ref{alg:fsg} in C++, using the implementation of the Fast Gauss Transform provided by Yang et al.\ \cite{ifgt}, and use the Python \texttt{SciPy}~\cite{scipy} and \texttt{stag}~\cite{stag} libraries
for eigenvector computation and graph operations respectively.
The \texttt{scikit-learn} and \texttt{FAISS} libraries are well-optimised and use C, C++, and FORTRAN for efficient implementation of core algorithms.
We first employ classical synthetic clustering datasets
to clearly compare how the running time of different algorithms is affected by the number of data points. This experiment highlights the nearly-linear time complexity of our algorithm. 
Next we  demonstrate the applicability of our new algorithm for image segmentation on the Berkeley Image Segmentation Dataset (BSDS)~\cite{amfm_pami2011}. 

For each experiment,
we set $k = 10$ for the approximate nearest neighbour algorithms.
For the synthetic datasets, we set the $\sigma$ value of the Gaussian kernel used by \textsc{SKLearn GK} and \textsc{Our Algorithm} to be $0.1$, and for the BSDS experiment we set $\sigma = 0.2$.
This choice follows the heuristic suggested by von Luxburg~\cite{Luxburg07} to choose $\sigma$ to approximate the distance from a point to its $k$th nearest neighbour.
All experiments are performed on an HP ZBook laptop with an 11th Gen Intel(R) Core(TM) i7-11800H @ 2.30GHz processor and 32 GB RAM.
The code to reproduce our results is available at \url{https://github.com/pmacg/kde-similarity-graph}.

\subsection{Synthetic Dataset}
In this experiment we use the \texttt{scikit-learn} library to generate 15,000 data points in $\mathbb{R}^2$ 
from a variety of classical synthetic clustering datasets.
The data is generated with the \texttt{make\_circles}, \texttt{make\_moons}, and \texttt{make\_blobs} methods of the \texttt{scikit-learn} library with a noise parameter of $0.05$.
The linear transformation $\{\{0.6, -0.6\}, \{-0.4, 0.8\}\}$ is applied to the blobs data to create asymmetric clusters.
As shown in Figure~\ref{fig:dataset_comparison}, all of the algorithms return approximately the same clustering, but our algorithm runs much faster than
the ones from \texttt{scikit-learn} and \texttt{FAISS}.

We further compare the speedup of our algorithm against the five competitors on the \texttt{two\_moons} dataset with an increasing number of data points, and our result is reported in Figure~\ref{fig:twomoons}. 
Notice that the running time of the  \texttt{scikit-learn} and \texttt{FAISS} algorithms scales roughly quadratically with the size of the input, while the running time of  our  algorithm is nearly linear.
Furthermore, we note that on a laptop with 32 GB RAM, the \textsc{SKLearn GK} algorithm could not scale beyond 20,000 data points due to its quadratic memory requirement, while
our new algorithm can cluster 1,000,000 data points in 240 seconds under the same memory constraint.

\begin{figure}[h]
\centering
\subfigure{
    \includegraphics[width=0.4\columnwidth]{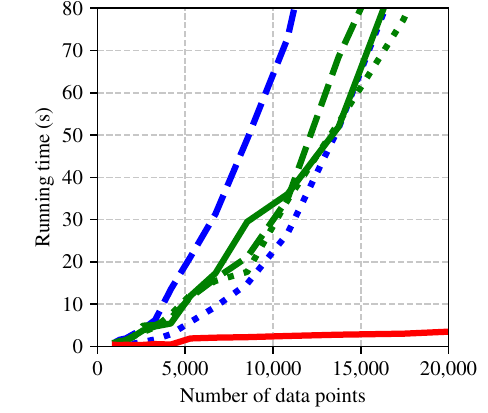}
}
\subfigure{
    \includegraphics[width=0.4\columnwidth]{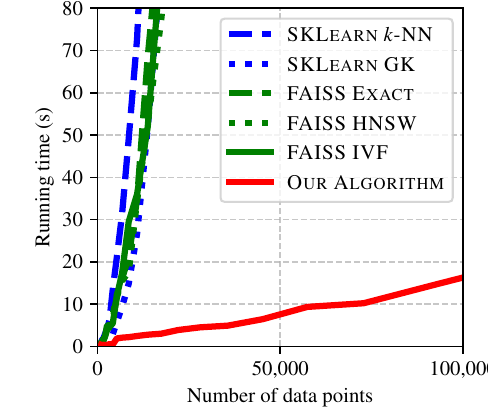}
}
\caption{Comparison of the running time of spectral clustering   on the two moons dataset. 
In every case, all algorithms return the correct clusters.
\label{fig:twomoons} }
\end{figure}

\subsection{BSDS Image Segmentation Dataset}
Finally we study  the application of spectral clustering for image segmentation on the BSDS dataset.
The dataset contains $500$ images with several ground-truth segmentations for each image.
Given an input image, we consider each pixel to be a point $(r, g, b, x, y)^\intercal \in \mathbb{R}^5$ where  $r$, $g$,  $b$ are the red, green,  blue pixel values and $x$, $y$
are the coordinates of the pixel in the image.
Then, we construct   a similarity graph based on  these points,  and apply spectral clustering, for which   we set the number of clusters to be the median number of clusters in the provided ground truth segmentations. Our experimental result is   reported  in Table~\ref{tab:bsds_results}, where the ``Downsampled'' version is employed to reduce the resolution of the image to 20,000 pixels in order to run  the \textsc{SKLearn GK} and the ``Full Resolution'' one is to apply the original image  of over 150,000 pixels  as input.  This set of experiments demonstrates our algorithm produces better clustering results with repsect to the average Rand Index~\cite{randindex}.



\begin{table}[h]
    \centering
    \caption{The average Rand Index of the output from  the six algorithms. Due to  the quadratic space complexity constraint,   the \textsc{SKLearn GK} cannot be applied to the full resolution image.
    \label{tab:bsds_results}}
    \begin{tabular}{ccc}
    \toprule
        Algorithm & Downsampled & Full Resolution \\ 
        \midrule
        \textsc{SKLearn GK} & 0.681  & -  \\
        \textsc{SKLearn $k$-NN} & 0.675  & 0.682  \\
        \textsc{FAISS Exact} & 0.675  & 0.680  \\
        \textsc{FAISS HNSW} & 0.679  & 0.677  \\
        \textsc{FAISS IVF} & 0.675  & 0.680  \\
        \textsc{Our Algorithm} & 0.680  & 0.702  \\
    \bottomrule
    \end{tabular}
\end{table}

\section{Conclusion}

In this paper we   develop a new technique that constructs   a similarity graph, and  show that an approximation algorithm for the \textsf{KDE} can be   employed   to  construct a similarity graph with proven approximation guarantee. Applying   
 the publicly   available implementations of the \textsf{KDE} as a black-box, our algorithm outperforms  five competing ones from the well-known \texttt{scikit-learn} and \texttt{FAISS} libraries  for the classical  datasets of low dimensions. We believe that our work will motivate more research on faster algorithms for the  \text{KDE} in higher dimensions and their efficient implementation, as well as more efficient constructions of similarity graphs.

\section*{Acknowledgements}

We would like to thank the anonymous reviewers for their  valuable comments on the paper.
This work is supported by an
EPSRC Early Career Fellowship (EP/T00729X/1).

\bibliographystyle{plain}

\bibliography{reference}

\newpage

\appendix

\appendix


\section{Proof of Theorem~\ref{thm:meta}} \label{sec:proof}

\renewcommand{\G}{\mathsf{K}}
\renewcommand{\graphh}{\mathsf{G}}
\renewcommand{\graphg}{\mathsf{K}}

This section presents   the complete proof of Theorem~\ref{thm:meta}. 
Let $y_{i, 1}, \ldots, y_{i, L}$ be random variables which are equal to the indices of the $L$ points sampled for $x_i$.
Recall that by the \textsf{SZ} algorithm, the ``ideal'' sampling probability for $x_j$ from $x_i$ is
\[
    p_i(j) \triangleq \min \left\{ C \cdot \frac{\log(n)}{\lambda_{k+1}}\cdot \frac{k(x_i, x_j)}{\deg_\G(v_i)} , 1 \right\}.
\]
We   denote the actual sampling probability that $x_j$ is sampled from $x_i$
under Algorithm~\ref{alg:fsg} to be
\[
    \widetilde{p}_i(j) \triangleq \mathbb{P}\left[x_j \in \{y_{i, 1}, \ldots y_{i, L}\}\right].
\]
Finally, for each added edge, Algorithm~\ref{alg:fsg} also computes an estimate of $p_i(x_j)$ which we denote 
\[
    \widehat{p}_i(j) \triangleq \min \left\{6 C\cdot  \frac{ \log(n)}{\lambda_{k+1}}\cdot \frac{k(x_i, x_j)}{g_{[1,n]}(x_i)}, 1 \right\}.
\]
Similar with the definition of $p_e$ in Section~\ref{sec:cps}, we define
\begin{itemize}
    \item $p(i, j) = p_i(j) + p_j(i) - p_i(j)\cdot p_j(i)$,
    \item $\widetilde{p}(i, j) = \widetilde{p}_i(j) + \widetilde{p}_j(i) - \widetilde{p}_i(j) \cdot\widetilde{p}_j(i)$, and 
    \item $\widehat{p}(i, j) = \widehat{p}_i(j) + \widehat{p}_j(i) - \widehat{p}_i(j)\cdot \widehat{p}_j(i)$.
\end{itemize}
Notice that, following the convention of \cite{SZ19}, we use $p_i(j)$ to refer to the probability that a given edge is sampled \emph{from the vertex $x_i$} and $p(i, j)$ is the probability that the given edge $\{v_i, v_j\}$ is sampled at all by the algorithm.
We use the same convention for $\widetilde{p}_i(j)$ and $\widehat{p}_i(j)$.

We first prove a sequence of lemmas showing that these probabilities are all within a constant factor of each other.

\begin{lemma} \label{lem:edgeprob}
For any point $x_i$, the probability that a given sampled neighbour $y_{i, l}$ is equal to $j$ is given by
\[
    \frac{k(x_i, x_j)}{2 \deg_\G(v_i)} \leq  \mathbb{P}\left[y_{i, l} = j\right]  \leq \frac{2 k(x_i, x_j)}{\deg_\G(v_i)}.
\]
\end{lemma}
\begin{proof}
Let
$X = \{x_1, \ldots, x_n\}$
be the input data points for Algorithm~\ref{alg:fsg}, and
$[n] = \{1, \ldots n\}$ be the indices of the input data points.
Furthermore, let $[a, b] = \{a, \ldots, b\}$ be the set of indices between $a$ and $b$.
Then, in each recursive call to Algorithm~\ref{alg:fsgsample}, we are given a range
of indices $[a, b]$
as input and assign $y_{i, l}$ to one half of it: either
$[a, \lfloor b/2 \rfloor]$ or $[\lfloor b/a \rfloor + 1, b]$.
By Algorithm~\ref{alg:fsgsample}, we have that the probability of assigning $y_{i, l}$ to
$[a, \lfloor b/2 \rfloor]$ is
\[
    \mathbb{P}\left[ y_{i, l} \in [a, \lfloor b/2 \rfloor ] ~|~ y_{i, l} \in [a, b] \right] = \frac{g_{[a, \lfloor b/2 \rfloor]}(x_i)}{g_{[a, b]}(x_i)}.
\]
By  the performance guarantee of the \textsc{KDE} algorithm, we have that $g_{[a, b]}(x_i) \in (1 \pm \epsilon) \deg_{[a, b]}(v_i)$, where we define \[
\deg_{[a, b]}(x_i) \triangleq \sum_{j = a}^b k(x_i, x_j).
\]
This gives 
\begin{equation} \label{eq:oneprob}
    \left(\frac{1 - \epsilon}{1 + \epsilon}\right) \frac{\deg_{[a, \lfloor b/2 \rfloor]}(v_i)}{\deg_{[a, b]}(v_i)} \leq \mathbb{P}\left[ y_{i,l} \in X_{[a, \lfloor b/2 \rfloor]} ~| y_{i,l} \in X_{[a, b]}\right] \leq \left(\frac{1 + \epsilon}{1 - \epsilon}\right) \frac{\deg_{[a, \lfloor b/2 \rfloor]}(v_i)}{\deg_{[a, b]}(v_i)}.
\end{equation}

Next, notice that we can write  for a sequence of intervals $[a_1,b_1]\subset[a_2, b_2]\subset\ldots\subset[1,n]$ that
\begin{align*}
   \mathbb{P}\left[ y_{i, l} = j \right] 
   & =  \mathbb{P}\left[ y_{i, l} = j | y_{i, l} \in [a_1, b_1] \right] \times \mathbb{P} \left[ y_{i, l} \in [a_1, b_1] | y_{i, l} \in [a_2, b_2] \right] \\
   & \qquad \qquad \times \ldots \times \mathbb{P} \left[ y_{i, l} \in [a_k, b_k] | y_{i, l} \in [1, n]\right],
\end{align*}
where each term corresponds to one level of recursion of Algorithm~\ref{alg:fsgsample} and there are at most $\lceil \log_2(n) \rceil$ terms.
Then, by \eqref{eq:oneprob}
and the fact that the denominator and numerator of adjacent terms cancel out, we have
\[
    \left(\frac{1 - \epsilon}{1 + \epsilon} \right)^{\lceil \log_2(n) \rceil} \frac{k(x_i, x_j)}{\deg_{\graphk}(v_i)}
    \leq \mathbb{P}\left[ y_{i, l} = j \right] \leq \left(\frac{1 + \epsilon}{1 - \epsilon}\right)^{\lceil \log_2(n)\rceil} \frac{k(x_i, x_j)}{\deg_{\graphk}(v_i)}
\]
since $\deg_{[j, j]}(v_i) = k(x_i, x_j)$ and $\deg_{[1, n]}(v_i) = \deg_{\graphk}(v_i)$.

For the lower bound, we have that
\begin{align*}
    \left(\frac{1 - \epsilon}{1 + \epsilon}\right)^{\lceil \log_2(n)\rceil} & \geq \left(1 - 2 \epsilon\right)^{\lceil \log_2(n)\rceil}  \geq 1 - 3 \log_2(n) \epsilon  \geq 1/2,
\end{align*}
where the final inequality follows by the condition of $\epsilon$ that $\epsilon \leq 1 / (6 \log_2(n))$. 

For the upper bound, we similarly have
\begin{align*}
   \left(\frac{1 + \epsilon}{1 - \epsilon}\right)^{\lceil \log_2(n)\rceil} & \leq  \left(1 + 3 \epsilon\right)^{\lceil \log_2(n)\rceil}  \leq \exp \left( 3 \lceil \log_2(n)\rceil \epsilon \right)  \leq \mathrm{e}^{2/3} \leq 2,
\end{align*}
where the first inequality follows since $\epsilon < 1 / (6\log_2(n))$.
\end{proof}

The next lemma shows that the probability of sampling each edge $\{v_i, v_j\}$ is approximately equal to the ``ideal'' sampling probability $p_i(j)$.

\begin{lemma} \label{lem:nodeprob}
    For every $i$ and $j \neq i$, we have
    \[
        \frac{9}{10}\cdot p_i(j) \leq \widetilde{p}_i(j) \leq 12\cdot p_i(j).
    \]
\end{lemma}
\begin{proof}
    Let $Y = \{y_{i, 1}, \ldots, y_{i, L}\}$ be the neighbours of $x_i$ sampled by Algorithm~\ref{alg:fsg}, where $L = 6 C \log(n) / \lambda_{k + 1}$.
    Then,
    \begin{align*}
        \mathbb{P} \left[j \in Y\right] & = 1 - \prod_{l = 1}^{L} \left(1 - \mathbb{P}\left[ y_{i, l} = j \right]  \right)   \geq 1 - \left(1 - \frac{k(x_i, x_j)}{2 \deg_{\graphk}(v_i)}\right)^L  \geq 1 - \exp\left(- L \cdot\frac{k(x_i, x_j)}{2 \deg_{\graphk}(v_i)}\right)
    \end{align*}

    The proof proceeds by case distinction. 

    \textbf{Case 1:} $p_i(j) \leq 9/10$. In this case, we have
    \begin{align*}
        \mathbb{P} \left[j \in Y\right] & \geq 1 - \exp\left(- 6 p_i(j) / 2\right)   \geq p_i(j).
    \end{align*}

    \textbf{Case 2:} $p_i(j) > 9/10$. In this case, we have
    \begin{align*}
        \mathbb{P} \left[j \in Y\right] & \geq 1 - \exp\left(- \frac{9 \cdot 6}{20} \right) \geq \frac{9}{10},
    \end{align*}
    which completes the proof on the lower bound of $\widetilde{p}_i(j)$.

    For the upper bound, we have
    \begin{align*}
        \mathbb{P}\left[j \in Y\right] & \leq 1 - \left(1 - \frac{2 k(x_i, x_j)}{\deg_{\graphk}(v_i)}\right)^L \leq \frac{2 k(x_i, x_j)}{\deg_{\graphk}(v_i)} \cdot L = 12C\cdot \frac{\log(n)}{\lambda_{k+1}}\cdot\frac{k(x_i, x_j)}{\deg_{\graphk}(v_i)}, 
    \end{align*}
    from which the statement follows.
\end{proof}
An immediate corollary of Lemma~\ref{lem:nodeprob} is as follows.
\begin{corollary} \label{cor:probs}
For all different $i, j \in [n]$, it holds that 
\[
    \frac{9}{10}\cdot  p(i, j) \leq \widetilde{p}(i, j) \leq 12 \cdot p(i, j)
\]
and
\[
    \frac{6}{7} \cdot p(i, j) \leq \widehat{p}(i, j) \leq \frac{36}{5} \cdot p(i, j).
\]
\end{corollary}
\begin{proof}
\newcommand{\tildep}{\widetilde{p}}
\newcommand{\hatp}{\widehat{p}}
    For the upper bound of the first statement, we can assume that $p_i(j) \leq 1/12$ and $p_j(i) \leq 1/12$, since otherwise we have $\tildep(i, j) \leq 1 \leq 12 \cdot p(i, j)$ and the statement holds trivially. 
    Then, by Lemma~\ref{lem:nodeprob}, we have
    \begin{align*}
        \tildep(i, j) 
        & = \tildep_i(j) + \tildep_j(i) - \tildep_i(j)\cdot \tildep_j(i) \\
        & \leq 12 p_i(j) + 12 p_j(i) - 12 p_i(j) \cdot 12 p_j(i) \\
        & \leq 12 \left(p_i(j) + p_j(i) - p_i(j) \cdot p_j(i)\right) \\
        & = 12 \cdot p(i,j)
    \end{align*}
    and 
    \begin{align*}
        \tildep(i, j) 
        & = \tildep_i(j) + \tildep_j(i) - \tildep_i(j)\cdot \tildep_j(i) \\
        & \geq \frac{9}{10}\cdot  p_i(j) + \frac{9}{10}\cdot  p_j(i) - \frac{9}{10} p_i(j) \cdot \frac{9}{10} p_j(i) \\
        & \geq \frac{9}{10} \left(p_i(j) + p_j(i) - p_i(j) p_j(i)\right) \\
        & = \frac{9}{10} \cdot p(i,j).
    \end{align*}
    For the second statement, notice that
    \begin{align*}
        \hatp_i(j)
        &= \min \left\{\frac{6 C \log(n)}{\lambda_{k+1}} \cdot \frac{k(i, j)}{g_{[1, n]}(x_i)}, 1  \right\} \\
        & \geq \min \left\{\frac{1}{1+\varepsilon} \frac{C \log(n)}{\lambda_{k+1}} \cdot \frac{k(i, j)}{\deg_{\graphk}(v_i)}, 1  \right\} \\
        & \geq \frac{1}{1+\varepsilon}\cdot \min \left\{\frac{C \log(n)}{\lambda_{k+1}} \cdot \frac{k(i, j)}{\deg_\graphk(v_i)}, 1  \right\} \\
        & = \frac{1}{1+\varepsilon}\cdot  p_i(j) \\
        & \geq \frac{6}{7}\cdot  p_i(j),
    \end{align*}
    since $g_{[1,n]}(x_i)$ is a $(1 \pm \varepsilon)$ approximation of $\deg_\graphk(v_i)$ and $\varepsilon \leq 1/6$.
    Similarly,
    \begin{align*}
        \hatp_i(j) \leq \frac{6}{1 - \varepsilon}\cdot  p_i(j) \leq \frac{36}{5}\cdot p_i(j).
    \end{align*}
    Then, for the upper bound of the second statement, we can assume that $p_i(j) \leq 5/36$ and $p_j(i) \leq 5/36$, since otherwise 
     $\hatp(i, j) \leq 1 \leq (36/5) \cdot \tildep(i, j)$ and the statement holds trivially. This implies that 
    \begin{align*}
        \hatp(i, j) 
        & = \hatp_i(j) + \hatp_j(i) - \hatp_i(j)\cdot \hatp_j(i) \\
        & \leq \frac{36}{5} p_i(j) + \frac{36}{5} p_j(i) - \frac{36}{5} p_i(j) \cdot \frac{36}{5} p_j(i) \\
        & \leq \frac{36}{5} \left(p_i(j) + p_j(i) - p_i(j)\cdot  p_j(i)\right) \\
        & = \frac{36}{5} \cdot p(i,j)
    \end{align*}
    and
    \begin{align*}
        \hatp(i, j) 
        & \geq \frac{6}{7} p_i(j) + \frac{6}{7} p_j(i) - \frac{6}{7} p_i(j) \cdot \frac{6}{7} p_j(i) \\
        & \geq \frac{6}{7} \left(p_i(j) + p_j(i) - p_i(j)\cdot  p_j(i)\right) \\
        & = \frac{6}{7} \cdot p(i,j),
    \end{align*}
    which completes the proof.
\end{proof}

We are now ready to prove  Theorem~\ref{thm:meta}. It is important to note that, although some of the proofs below are parallel to that of \cite{SZ19}, our analysis needs to carefully take into account the approximation ratios introduced by the 
approximate \textsc{KDE} algorithm, which  makes our analysis more involved. 
The following concentration inequalities will be used in our proof.

\begin{lemma}[Bernstein's Inequality~\cite{chung2006concentration}]
Let $X_1, \ldots, X_n$ be independent random variables such that $|X_i| \leq M$ for any $i \in \{1, \ldots, n\}$.
Let $X = \sum_{i = 1}^n X_i$, and  $R = \sum_{i = 1}^n \mathbb{E}\left[ X_i^2 \right]$.
Then, it holds that
\[
    \mathbb{P}\left[|X - \mathbb{E}\left[ X \right] | \geq t\right] \leq 2 \exp\left(-\frac{t^2}{2(R + Mt/3)}\right).
\]
\end{lemma}

\begin{lemma}[Matrix Chernoff Bound~\cite{tropp2012user}]
    Consider a finite sequence $\{X_i\}$ of independent, random, PSD matrices of dimension $d$ that satisfy $\|X_i\| \leq R$.
    Let $\mu_{\mathrm{min}} \triangleq \lambda_{\mathrm{min}}(\mathbb{E}\left[ \sum_i X_i\right])$
    and $\mu_{\mathrm{max}} \triangleq \lambda_{\mathrm{max}}(\mathbb{E}\left[ \sum_i X_i\right])$.
    Then, it holds that
    \[
        \mathbb{P}\left[ \lambda_{\mathrm{min}}\left(\sum_i X_i\right) \leq (1 - \delta) \mu_{\mathrm{min}}  \right] \leq d \left( \frac{\mathrm{e}^{-\delta}}{(1 - \delta)^{1-\delta}}\right)^{\mu_{\mathrm{min}} / R}
    \]
    for $\delta \in [0, 1]$, and
    \[
        \mathbb{P}\left[ \lambda_{\mathrm{max}}\left(\sum_i X_i\right) \geq (1 + \delta) \mu_{\mathrm{max}}  \right] \leq d \left( \frac{\mathrm{e}^{\delta}}{(1 + \delta)^{1+\delta}}\right)^{\mu_{\mathrm{max}} / R}
    \]
    forr $\delta \geq 0$.
\end{lemma}

\begin{proof}[Proof of Theorem~\ref{thm:meta}]
We first show that the degrees of all the vertices  in the similarity graph $\graphg$ are preserved with high probability in the sparsifier $\graphh$. 
For any vertex  $v_i$, and   let $y_{i,1}, \ldots, y_{i,L}$ be the indices of the neighbours of $v_i$ sampled by Algorithm~\ref{alg:fsg}.

For every pair of indices $i \neq j$, and for every $1 \leq l \leq L$, we define the random variable $Z_{i, j, l}$ to be the weight of the sampled edge if $j$ is the neighbour sampled from $i$ at iteration $l$, and $0$ otherwise:
\[
    Z_{i, j, l} \triangleq \twopartdefow{\frac{k(x_i, x_j)}{\widehat{p}(i, j)}}{y_{i, l} = j}{0}
\]
Then, fixing an arbitrary vertex $x_i$, we can write
\[
\deg_\graphh(v_i) = \sum_{l = 1}^L \sum_{i \neq j} Z_{i, j, l} + Z_{j, i, l}.
\]
We have
\begin{align*}
\mathbb{E}\left[ \deg_\graphh(v_i) \right] & = \sum_{l = 1}^L \sum_{j \neq i} \mathbb{E}\left[ Z_{i, j, l}\right] +  \mathbb{E}\left[ Z_{j, i, l} \right] \\
& = \sum_{l = 1}^L \sum_{j \neq i} 
\left[ \mathbb{P}\left[y_{i, l} = j\right] \cdot \frac{k(x_i, x_j)}{\widehat{p}(i, j)} + \mathbb{P}\left[y_{j, l} = i\right] \cdot \frac{k(x_i, x_j)}{\widehat{p}(i, j)} \right].
\end{align*}
By Lemmas~\ref{lem:edgeprob} and \ref{lem:nodeprob} and Corollary~\ref{cor:probs}, we have
\begin{align*}
\mathbb{E}\left[ \deg_\graphh(v_i) \right] & \geq \sum_{j \neq i} \frac{k(x_i, x_j)}{\widehat{p}(i, j)} \left( \frac{L \cdot k(x_i, x_j)}{2 \deg_\graphk(v_i)} + \frac{L \cdot k(x_i, x_j)}{2 \deg_{\graphk}(v_j)} \right) \\
& = \sum_{i \neq j} \frac{3 k(x_i, x_j)}{\widehat{p}(i, j)} \left(p_i(j) + p_j(i)\right) \\
& \geq \sum_{i \neq j} \frac{5 k(x_i, x_j)}{12} = \frac{5 \deg_\graphk(v_i)}{12},
\end{align*}
where the last inequality follows by the fact that 
$
\widehat{p}(i,j) \leq (36/5) \cdot  p(i,j)  \leq (36/5) \cdot \left( p_i(j) + p_j(i)\right)
$.
Similarly, we have
\begin{align*}
\mathbb{E}\left[ \deg_\graphh(v_i) \right] & \leq \sum_{j \neq i} \frac{k(x_i, x_j)}{\widehat{p}(i, j)} \left( \frac{2 \cdot L \cdot k(x_i, x_j)}{\deg_\graphk(v_i)} + \frac{2 \cdot L \cdot k(x_i, x_j)}{\deg_\graphk(v_j)} \right) \\
& = \sum_{j \neq i} \frac{12 \cdot k(x_i, x_j)}{\widehat{p}(i, j)} \left(p_i(j)  + p_j(i) \right) \\
& \leq \sum_{j \neq i} 28 \cdot k(x_i, x_j) = 28 \cdot \deg_\G(v_i),
\end{align*}
where the inequality follows by the fact that
\[
\widehat{p}(i,j) \geq \frac{6}{7} \cdot p(i,j) = \frac{6}{7}\cdot\left( p_i(j) + p_j(i)-  p_i(j) \cdot  p_j(i)\right) \geq \frac{3}{7}\cdot\left( p_i(j) + p_j(i)\right).
\]
In order to prove a concentration bound on this degree estimate, we would like to apply the Bernstein inequality for which we need to bound
\begin{align*}
    R & = \sum_{l = 1}^L \sum_{j \neq i} \mathbb{E}\left[ Z_{i,j,l}^2\right] + \mathbb{E}\left[ Z_{j,i, l}^2\right] \\
    & = \sum_{l = 1}^L \sum_{j \neq i} \mathbb{P}\left[y_{i, l} = j\right] \cdot \frac{k(x_i, x_j)^2}{\widehat{p}(i, j)^2} + \mathbb{P}\left[y_{j, l} = i\right] \cdot \frac{k(x_i, x_j)^2}{\widehat{p}(i, j)^2} \\
    & \leq \sum_{j \neq i} \frac{12 \cdot k(x_i, x_j)^2}{\widehat{p}(i, j)^2}\cdot  \left(p_i(j) + p_j(i)\right) \\
    & \leq \sum_{j \neq i} 28 \cdot \frac{k(x_i, x_j)^2}{\widehat{p}(i, j)} \\
    & \leq \sum_{j \neq i} \frac{98}{3} \cdot \frac{k(x_i, x_j)^2}{p_i(j)} \\
    & = \sum_{j \neq i} \frac{98}{3} \cdot \frac{k(x_i, x_j)\cdot \deg_\G(v_i)\cdot  \lambda_{k + 1}}{C \log(n)} \\
    & = \frac{98 \cdot \deg_\G(v_i)^2 \cdot \lambda_{k + 1}}{3 \cdot C \log(n)},
\end{align*}
where the third inequality follows by the fact that 
\[
\widehat{p}(i,j) \geq \frac{6}{7}\cdot p(i,j) \geq \frac{6}{7}\cdot  p_i(j).
\]
Then, by applying Bernstein's inequality we have for any constant $C_2$ that
\begin{align*}
    \mathbb{P}\left[ \left|\deg_\graphh(v_i) - \mathbb{E}[\deg_\graphh(v_i)]\right| \geq \frac{1}{C_2} \deg_\graphg(v_i) \right] & \leq 2 \exp\left(- \frac{\deg_{\graphg}(v_i)^2 / (2\cdot C_2^2)}{\frac{98  \deg_\graphg(v_i)^2 \lambda_{k+1}}{3 C \log(n)} +  \frac{7 \deg_\graphg(v_i)^2 \lambda_{k+1}}{6 C C_2 \cdot \log(n)}}  \right) \\
    & \leq 2 \exp\left( - \frac{C\cdot  \log(n)}{((196/3) \cdot C_2^2 + (7/3) \cdot C_2)\cdot \lambda_{k+1}}  \right) \\
    & = o(1/n),
\end{align*}
where we use the fact that
\[
Z_{i, j, l} \leq \frac{7 k(x_i, x_j)}{6 p_i(j)} = \frac{7 \deg_\graphk(v_i)\cdot  \lambda_{k+1}}{6 C \cdot \log(n)}.
\]
Therefore, by taking $C$ to be sufficiently large and by the union bound, it holds with high probability that the degree of all the nodes in $\graphh$ are preserved up to a constant factor.
For the remainder of the proof, we   assume that this is the case.
Note in particular that this implies $\vol_\graphh(S) = \Theta(\vol_\graphg(S))$ for any subset $S \subseteq V$.

Next, we prove it holds for $\graphh$ that  
$    \phi_\graphh(S_i) = O\left(k \cdot \phi_\G(S_i)\right)$
for any $1 \leq i \leq k$, where $S_1, \ldots,S_k$ form an optimal clustering in $\G$.

By the definition of $Z_{i, j, l}$, it holds for any $1 \leq i \leq k$ that
\begin{align*}
\mathbb{E}\left[\partial_\graphh(S_i)\right] & = \mathbb{E}\left[ \sum_{a \in S_i} \sum_{b \not\in S_i} \sum_{l = 1}^L Z_{a, b, l} + Z_{b, a, l} \right] \\
    & \leq \sum_{a \in S_i} \sum_{b \not\in S_i} \frac{12 k(x_a, x_b)}{\widehat{p}(a, b)}\cdot \left(p_a(b) + p_b(a)\right) \\
    & = O\left(\partial_\G(S_i)\right)
\end{align*}
where the last line follows by Corollary~\ref{cor:probs}.
By Markov's inequality and the union bound, with constant probability it holds for all $i = 1, \ldots, k$ that
\[
    \partial_\graphh(S_i) = O(k \cdot \partial_\G(S_i)).
\]
Therefore, it holds with constant probability that
\[
    \rho_\graphh(k) \leq \max_{1 \leq i \leq k} \phi_\graphh(S_i) = \max_{1 \leq i \leq k} O(k \cdot \phi_\G(S_i)) = O(k \cdot \rho_\G(k)).
\]
Finally, we prove that 
$\lambda_{k+1}(\mathbf{N}_\graphh) = \Omega(\lambda_{k+1}(\mathbf{N}_\G))$. Let $\overline{\mathbf{N}}_\G$ be the projection of $\mathbf{N}_\G$ on its top $n - k$ eigenspaces, and notice that $\overline{\mathbf{N}}_\G$ can be written
\[
    \overline{\mathbf{N}}_\G = \sum_{i = k+1}^n \lambda_i(\mathbf{N}_\graphk) f_i f_i^\intercal
\]
where $f_1, \ldots, f_n$ are the eigenvectors of $\mathbf{N}_\G$.
Let $\overline{\mathbf{N}}_\G^{-1/2}$ be the square root of the pseudoinverse of $\overline{\mathbf{N}}_\G$.

We   prove that the top $n - k$ eigenvalues of $\mathbf{N}_\G$ are preserved, which implies that $\lambda_{k+1}(\mathbf{N}_\mathsf{G}) = \Theta(\lambda_{k+1}(\mathbf{N}_\graphk))$.
To prove this, for each data point $x_i$ and sample $1 \leq l \leq L$, we define a random matrix $X_{i, l} \in \mathbb{R}^{n \times n}$ by
\[
    X_{i, l} = w_{\graphh}(v_i, v_j) \cdot \overline{\mathbf{N}}_\G^{-1/2} b_e b_e^\intercal \overline{\mathbf{N}}_\G^{-1/2};
\]
where $j = y_{i, l}$,  
$b_e \triangleq \chi_{v_i} - \chi_{v_j}$ is the  edge indicator vector, and 
 $x_{v_i}\in\mathbb{R}^n$ is defined by
\[
    \chi_{v_i}(a) \triangleq \twopartdefow{\frac{1}{\sqrt{\deg_\graphk(v_i)}}}{a = i}{0}
\]

Notice that
\[
    \sum_{i = 1}^n \sum_{l = 1}^L X_{i, l} = \sum_{\text{sampled edges } e = \{v_i, v_j\}} w_{\graphh}(v_i, v_j) \cdot \overline{\mathbf{N}}_\G^{-1/2} b_e b_e^\intercal \overline{\mathbf{N}}_\G^{-1/2} = \overline{\mathbf{N}}_\G^{-1/2} \mathbf{N}_\graphh^{'} \overline{\mathbf{N}}_\G^{-1/2}
\]
where
\[
    \mathbf{N}_\graphh^{'} = \sum_{\text{sampled edges } e =\{v_i,v_j\}} w_{\graphh}(v_i, v_j) \cdot b_e b_e^\intercal
\]
is the Laplacian matrix of $\graphh$ normalised with respect to the degrees of the nodes in $\G$. 
We   prove that, with high probability, the top $n - k$ eigenvectors of $\mathbf{N}_\graphh^{'}$ and $\mathbf{N}_\G$ are approximately the same.
Then, we   show the same for $\mathbf{N}_\graphh$ and $\mathbf{N}_\graphh^{'}$ which implies that $\lambda_{k+1}(\mathbf{N}_\graphh) = \Omega(\lambda_{k+1}(\mathbf{N}_\G))$.

We begin by looking at the first moment of $\sum_{i = 1}^n \sum_{l = 1}^L X_{i, l}$, and have that  
\begin{align*}
    \lambda_{\mathrm{min}} \left(\mathbb{E}\left[ \sum_{i = 1}^n \sum_{l = 1}^L X_{i, l} \right]\right) & = \lambda_{\mathrm{min}}\left(\sum_{i = 1}^n \sum_{l = 1}^L \sum_{\substack{j\neq i\\  e= \{v_i, v_j\}} } \mathbb{P}\left[y_{i, l} = j\right] \cdot \frac{k(x_i, x_j)}{\widehat{p}(i, j)} \cdot \overline{\mathbf{N}}_\G^{-1/2} b_e b_e^\intercal \overline{\mathbf{N}}_\G^{-1/2}\right) \\
     & \geq \lambda_{\mathrm{min}}\left(\sum_{i = 1}^n  \sum_{\substack{j\neq i\\  e= \{v_i, v_j\} }}3 p_i(j) \cdot \frac{k(x_i, x_j)}{\widehat{p}(i, j)} \cdot \overline{\mathbf{N}}_\G^{-1/2} b_e b_e^\intercal \overline{\mathbf{N}}_\G^{-1/2}\right) \\
     & \geq \lambda_{\mathrm{min}} \left(\frac{5}{12}\cdot \overline{\mathbf{N}}_\G^{-1/2} \mathbf{N}_\G \overline{\mathbf{N}}_\G^{-1/2} \right) = \frac{5}{12},
\end{align*}
where the last inequality follows by the fact that
\[
\widehat{p}(i,j) \leq \frac{36}{5} \cdot p(i,j) \leq \frac{36}{5} \cdot \left(p_i(j) + p_j(i) \right).
\]
Similarly,
\begin{align*}
    \lambda_{\mathrm{max}} \left(\mathbb{E}\left[ \sum_{i = 1}^n \sum_{l = 1}^L X_{i, l} \right]\right) & = \lambda_{\mathrm{max}}\left(\sum_{i = 1}^n \sum_{l = 1}^L\sum_{\substack{j\neq i\\  e= \{v_i, v_j\}} } 
    \mathbb{P}\left[y_{i, l} = j \right] \cdot\frac{k(x_i, x_j)}{\widehat{p}(i, j)} \cdot \overline{\mathbf{N}}_\G^{-1/2} b_e b_e^\intercal \overline{\mathbf{N}}_\G^{-1/2}\right) \\
     & \leq \lambda_{\mathrm{max}}\left(\sum_{i = 1}^n \sum_{\substack{j\neq i\\  e= \{v_i, v_j\} }}12 \cdot p_i(j)\cdot  \frac{k(x_i, x_j)}{\widehat{p}(i, j)} \cdot \overline{\mathbf{N}}_\G^{-1/2} b_e b_e^\intercal \overline{\mathbf{N}}_\G^{-1/2}\right) \\
     & \leq \lambda_{\mathrm{max}} \left(28 \cdot \overline{\mathbf{N}}_\G^{-1/2} \mathbf{N}_\G \overline{\mathbf{N}}_\G^{-1/2} \right) = 28,
\end{align*}
where the last inequality follows by the fact that\[
\widehat{p}(i,j) \geq \frac{6}{7}\cdot p(i,j) \geq \frac{3}{7}\cdot (p_i(j) + p_j(i)).
\]
Additionally, for any $i$ and $j = y_{i, l}$ and $e=\{v_i, v_j\}$, we have that
\begin{align*}
    \| X_{i, l} \| & \leq w_\graphh(v_i, v_j) \cdot b_e^\intercal \overline{\mathbf{N}}_\G^{-1/2} \overline{\mathbf{N}}_\G^{-1/2} b_e \\
    & = \frac{k(x_i, x_j)}{\widehat{p}(i, j)} \cdot b_e^\intercal \overline{\mathbf{N}}_\G^{-1} b_e \\
    & \leq \frac{k(x_i, x_j)}{\widehat{p}(i, j)} \cdot \frac{1}{\lambda_{k+1}} \| b_e \|^2 \\
    & \leq \frac{7\cdot  \lambda_{k+1}}{3 C \log(n) \left(\frac{1}{\deg_\G(v_i)} + \frac{1}{\deg_\G(v_j)}\right)} \cdot \frac{1}{\lambda_{k+1}} \left( \frac{1}{\deg_\G(v_i)} + \frac{1}{\deg_\G(v_j)} \right) \\
    & \leq \frac{7}{3 C \log(n)}.
\end{align*}
Now, we apply the matrix Chernoff bound and have that 
\begin{align*}
    \mathbb{P}\left[ \lambda_{\mathrm{max}}\left(\sum_{i = 1}^n \sum_{l = 1}^L X_{i, l}\right) \geq 42 \right] \leq n \left( \frac{e^{1/2}}{(1+1/2)^{3/2}} \right)^{ 12 C \cdot \log(n)} = O(1/n^c)
\end{align*}
for some constant $c$.
The other side of the matrix Chernoff bound gives us that
\begin{align*}
    \mathbb{P}\left[ \lambda_{\mathrm{min}}\left(\sum_{i = 1}^n \sum_{l = 1}^L X_{i, l}\right) \leq 5/24 \right] \leq O(1/n^c).
\end{align*}
Combining these, with probability $1 - O(1/n^c)$ it holds for any non-zero $x \in \mathbb{R}^n$ in the space spanned by $f_{k+1}, \ldots, f_n$ that
\[
    \frac{x^\intercal \overline{\mathbf{N}}_\G^{-1/2} \mathbf{N}_\graphh^{'} \overline{\mathbf{N}}_\G^{-1/2}x}{x^\intercal x} \in (5/24, 42).
\]
By setting $y = \overline{\mathbf{N}}_\G^{-1/2}x$, we can rewrite this as
\[
    \frac{y^\intercal \mathbf{N}_\graphh^{'} y}{y^\intercal \overline{\mathbf{N}}_\G^{1/2} \overline{\mathbf{N}}_\G^{1/2} y} = \frac{y^\intercal \mathbf{N}_\graphh^{'} y}{y^\intercal \overline{\mathbf{N}}_\G y} = \frac{y^\intercal \mathbf{N}_\graphh^{'} y}{y^\intercal y} \frac{y^\intercal y}{y^\intercal \overline{\mathbf{N}}_\graphg y} \in (5/24, 42).
\]
Since $\mathrm{dim}(\mathrm{span}\{f_{k+1}, \ldots, f_n\}) = n - k$, we have shown that there exist $n - k$ orthogonal vectors whose Rayleigh quotient with respect to $\mathrm{N}_\graphh^{'}$ is $\Omega(\lambda_{k+1}(\mathbf{N}_\G))$.
By the Courant-Fischer Theorem, we have
$    \lambda_{k+1}(\mathbf{N}_\graphh^{'}) = \Omega(\lambda_{k+1}(\mathbf{N}_\G))$.

It only remains to show that $\lambda_{k+1}(\mathbf{N}_\graphh) = \Omega(\lambda_{k+1}(\mathbf{N}_\graphh^{'}))$, which implies that $\lambda_{k+1}(\mathbf{N}_\graphh) = \Omega(\lambda_{k+1}(\mathbf{N}_\G))$.
By the definition of $\mathbf{N}_\graphh^{'}$, we have that $\mathbf{N}_\graphh = \mathbf{D}_\graphh^{-1/2} \mathbf{D}_\G^{1/2} \mathbf{N}_\graphh^{'} \mathbf{D}_\G^{1/2} \mathbf{D}_\graphh^{-1/2}$.
Therefore, for any $x \in \mathbb{R}^n$ and $y = \mathbf{D}_\G^{1/2}\mathbf{D}_\graphh^{-1/2} x$, it holds that
\[
    \frac{x^\intercal \mathbf{N}_\graphh x}{x^\intercal x} = \frac{y^\intercal \mathbf{N}_\graphh^{'} y}{x^\intercal x} = \Omega\left(\frac{y^\intercal \mathbf{N}_\graphh^{'} y}{y^\intercal y}\right),
\]
where the final guarantee follows from the fact that the degrees in $\graphh$ are preserved up to a constant factor.
The conclusion of the theorem follows by the Courant-Fischer Theorem.

Finally, we bound the running time of Algorithm~\ref{alg:fsg} which is dominated
by the recursive calls to Algorithm~\ref{alg:fsgsample}.
We note that, although the number of nodes doubles at each level of the recursion tree (visualised in Figure~\ref{fig:algtree}), the total number of samples $S$ and data points $X$ remain constant for each level of the tree.
Then, since the running time of the \textsf{KDE} algorithm is superadditive, the total running time of the \textsf{KDE} algorithms at level $i$ of the tree is
\begin{align*}
    T_{i} & = \sum_{j = 1}^{2^i} T_{\textsf{KDE}}(|S_{i, j}|, |X_{i, j}|, \epsilon) \\
    & \leq T_{\textsf{KDE}}\left(\sum_{j = 1}^{2^i} |S_{i, j}|, \sum_{j = 1}^{2^i} |X_{i, j}|, \epsilon \right) = T_{\textsf{KDE}}(|S|, |X|, \epsilon).
\end{align*}
Since there are $O(\log(n))$ levels of the tree, the total running time of Algorithm~\ref{alg:fsgsample} is $\widetilde{O}(T_{\textsf{KDE}}(|S|, |X|, \epsilon) )$. This  completes the proof.
\end{proof}

\section{Additional Experimental Results} \label{app:experiments}
In this section, we include in Figures~\ref{fig:bsds_results_app1}~and~\ref{fig:bsdsapp2} some additional examples of the performance of the six spectral clustering algorithms on the BSDS image segmentation dataset.
Due to its quadratic memory requirement, 
the \textsc{SKLearn GK} algorithm  cannot be used on the full-resolution image. Therefore, we present its results on each image downsampled to 20,000 pixels.
For every other algorithm, we show the results on the full-resolution image.
In every case, we find that our algorithm is able to identify more refined detail of the image when compared with the alternative algorithms.

\begin{figure}[htp]
    \centering
    \subfigure[Original Image] {\includegraphics[width=0.22\columnwidth]{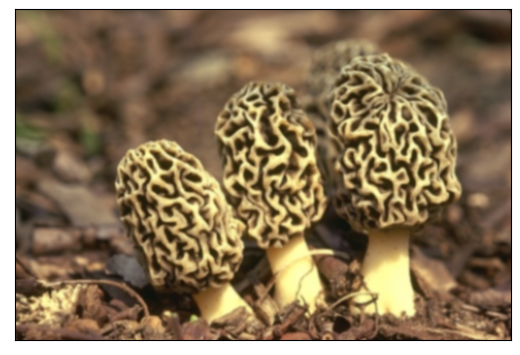}}
    \hspace{1em}
    \subfigure[\textsc{SKLearn GK}] {\includegraphics[width=0.22\columnwidth]{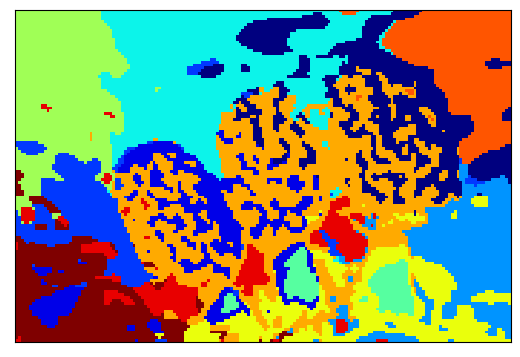}}
    \hspace{1em}
    \subfigure[\textsc{Our Algorithm}] {\includegraphics[width=0.22\columnwidth]{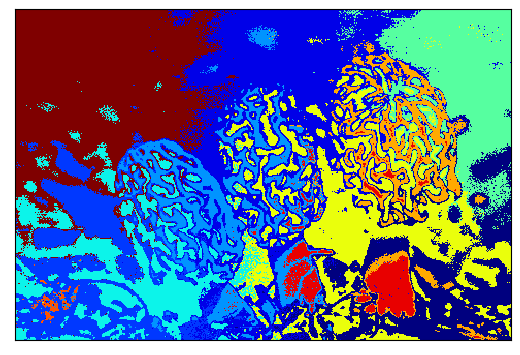}}
    \\
    \subfigure[\textsc{SKLearn $k$-NN}] {\includegraphics[width=0.22\columnwidth]{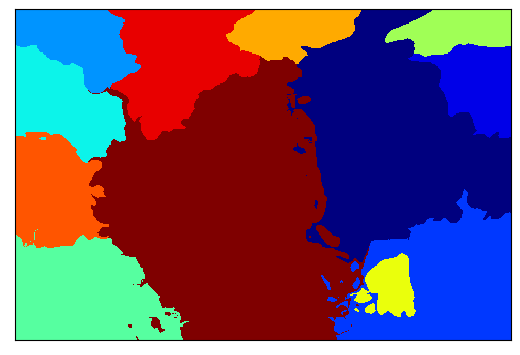}}
    \hspace{1em}
    \subfigure[\textsc{FAISS Exact}] {\includegraphics[width=0.22\textwidth]{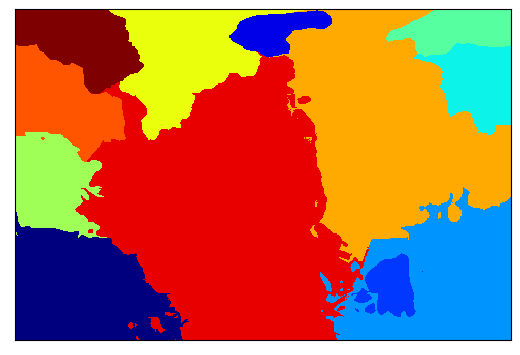}}
    \hspace{1em}
    \subfigure[\textsc{FAISS HNSW}] {\includegraphics[width=0.22\textwidth]{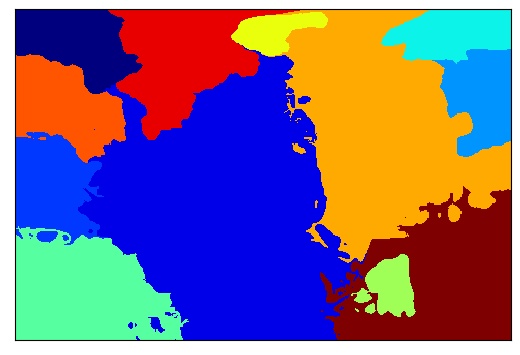}}
    \hspace{1em}
    \subfigure[\textsc{FAISS IVF}] {\includegraphics[width=0.22\textwidth]{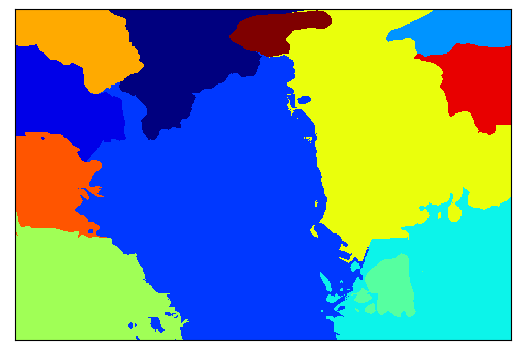}}
    \\
    \par\noindent\rule{\textwidth}{0.5pt}
    \subfigure[Original Image] {\includegraphics[width=0.22\columnwidth]{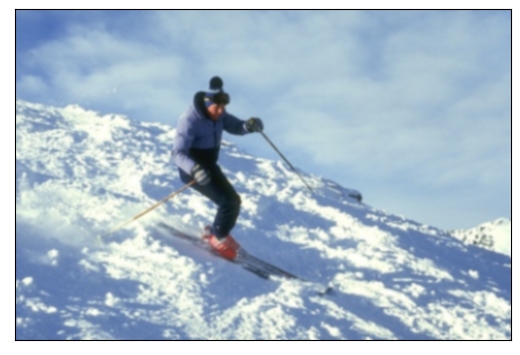}}
    \hspace{1em}
    \subfigure[\textsc{SKLearn GK}] {\includegraphics[width=0.22\columnwidth]{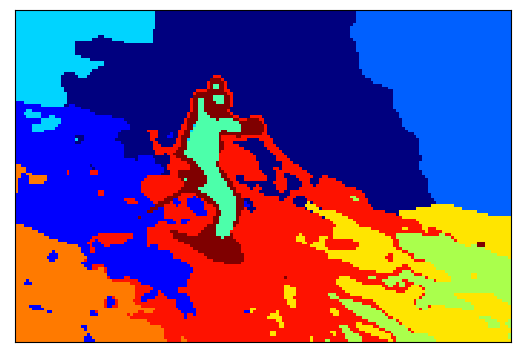}}
    \hspace{1em}
    \subfigure[\textsc{Our Algorithm}] {\includegraphics[width=0.22\columnwidth]{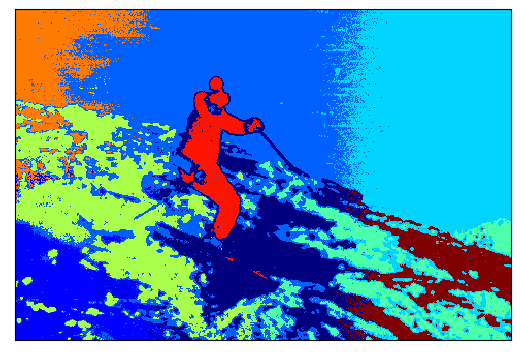}}
    \\
    \subfigure[\textsc{SKLearn $k$-NN}] {\includegraphics[width=0.22\columnwidth]{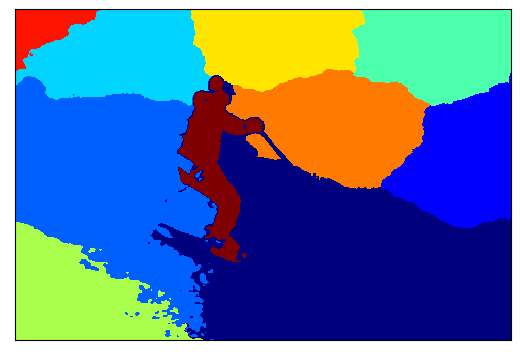}}
    \hspace{1em}
    \subfigure[\textsc{FAISS Exact}] {\includegraphics[width=0.22\textwidth]{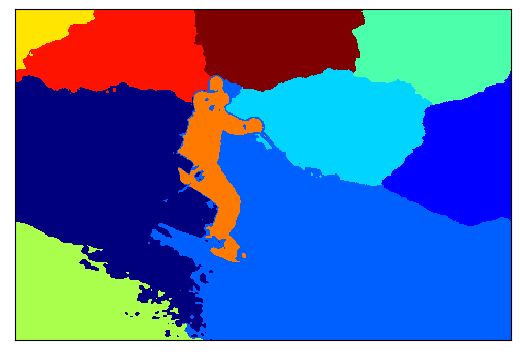}}
    \hspace{1em}
    \subfigure[\textsc{FAISS HNSW}] {\includegraphics[width=0.22\textwidth]{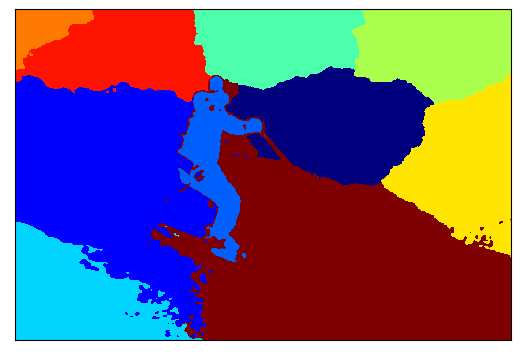}}
    \hspace{1em}
    \subfigure[\textsc{FAISS IVF}] {\includegraphics[width=0.22\textwidth]{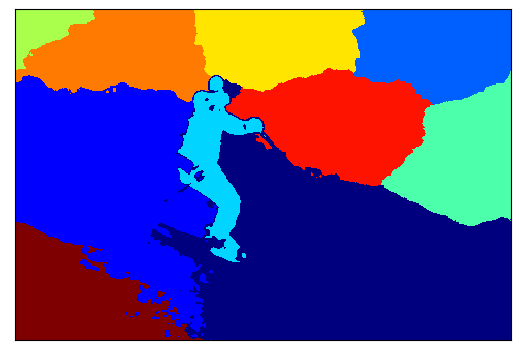}}
    \\
    \par\noindent\rule{\textwidth}{0.5pt}
    \subfigure[Original Image] {\includegraphics[width=0.22\columnwidth]{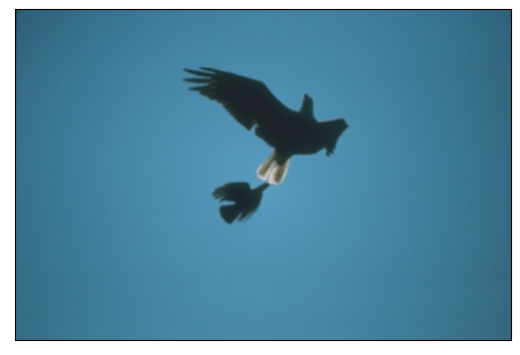}}
    \hspace{1em}
    \subfigure[\textsc{SKLearn GK}] {\includegraphics[width=0.22\columnwidth]{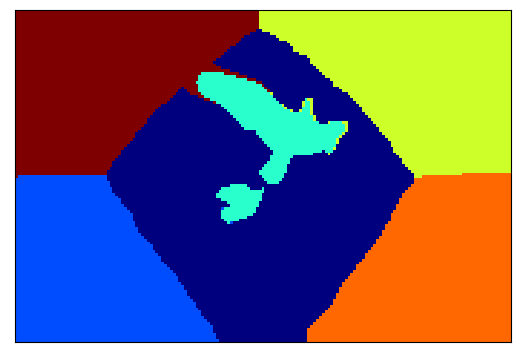}}
    \hspace{1em}
    \subfigure[\textsc{Our Algorithm}] {\includegraphics[width=0.22\columnwidth]{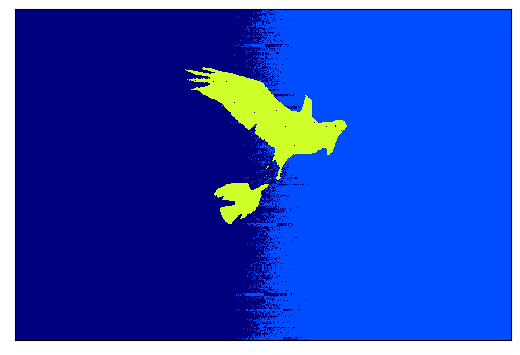}}
    \\
    \subfigure[\textsc{SKLearn $k$-NN}] {\includegraphics[width=0.22\columnwidth]{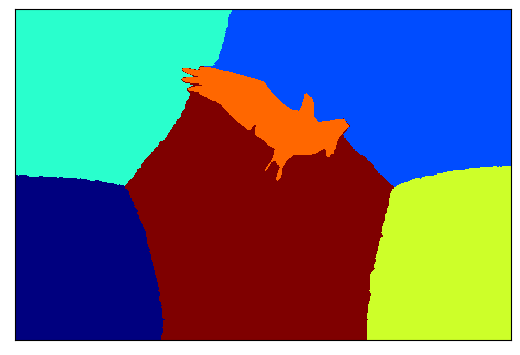}}
    \hspace{1em}
    \subfigure[\textsc{FAISS Exact}] {\includegraphics[width=0.22\textwidth]{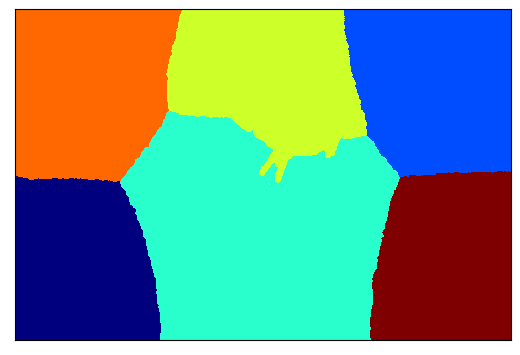}}
    \hspace{1em}
    \subfigure[\textsc{FAISS HNSW}] {\includegraphics[width=0.22\textwidth]{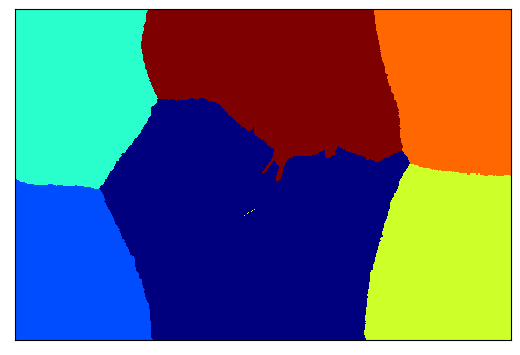}}
    \hspace{1em}
    \subfigure[\textsc{FAISS IVF}] {\includegraphics[width=0.22\textwidth]{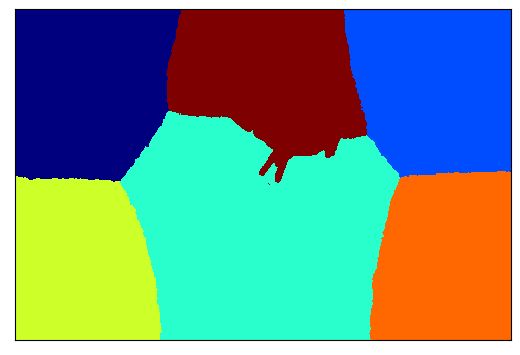}}

    

    
    \caption{More examples on the performance of the spectral clustering algorithms for image segmentation.
    \label{fig:bsds_results_app1}
    }   
\end{figure}

\begin{figure}[htp]
\centering
    \subfigure[Original Image] {\includegraphics[width=0.2\columnwidth]{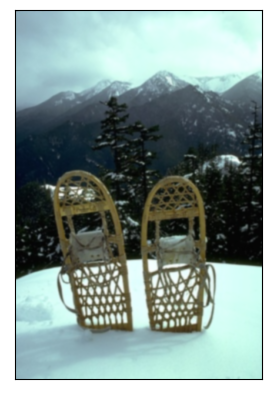}}
    \hspace{1em}
    \subfigure[\textsc{SKLearn GK}] {\includegraphics[width=0.2\columnwidth]{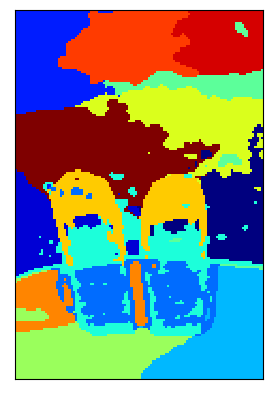}}
    \hspace{1em}
    \subfigure[\textsc{Our Algorithm}] {\includegraphics[width=0.2\columnwidth]{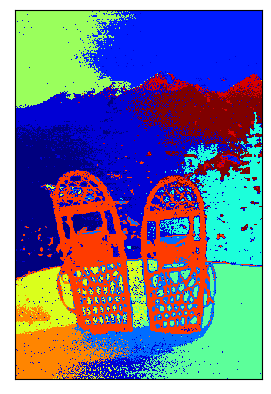}}
    \\
    \subfigure[\textsc{SKLearn $k$-NN}] {\includegraphics[width=0.2\columnwidth]{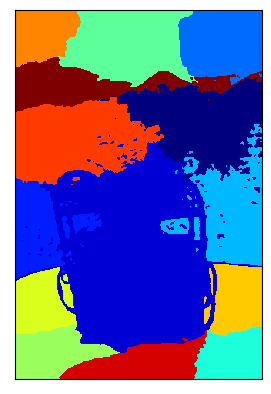}}
    \hspace{1em}
    \subfigure[\textsc{FAISS Exact}] {\includegraphics[width=0.2\textwidth]{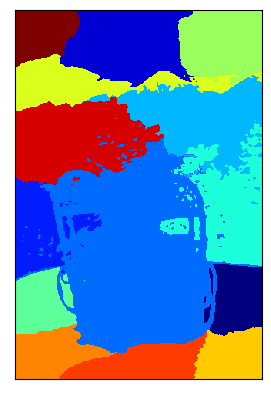}}
    \hspace{1em}
    \subfigure[\textsc{FAISS HNSW}] {\includegraphics[width=0.2\textwidth]{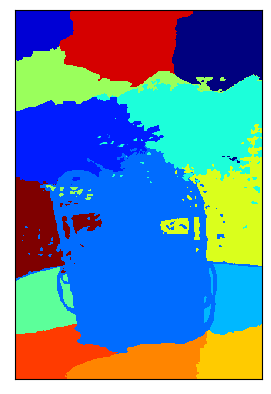}}
    \hspace{1em}
    \subfigure[\textsc{FAISS IVF}] {\includegraphics[width=0.2\textwidth]{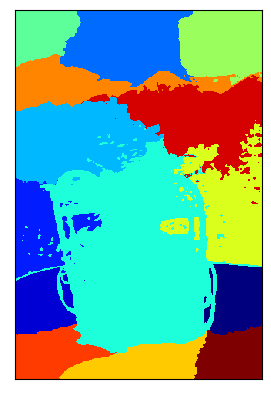}}
    \\
    \par\noindent\rule{\textwidth}{0.5pt}
    
    \subfigure[Original Image] {\includegraphics[width=0.2\columnwidth]{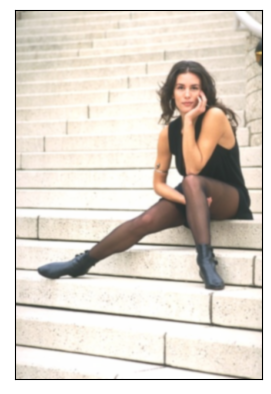}}
    \hspace{1em}
    \subfigure[\textsc{SKLearn GK}] {\includegraphics[width=0.2\columnwidth]{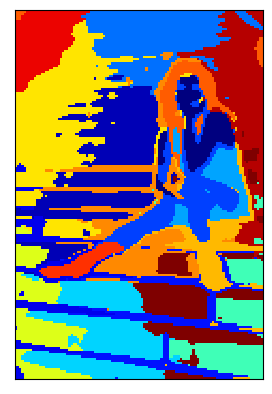}}
    \hspace{1em}
    \subfigure[\textsc{Our Algorithm}] {\includegraphics[width=0.2\columnwidth]{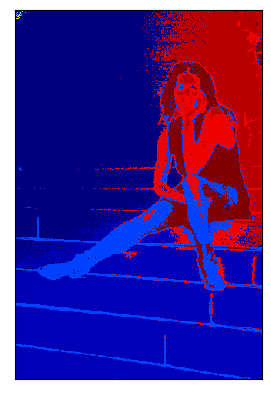}}
    \\
    \subfigure[\textsc{SKLearn $k$-NN}] {\includegraphics[width=0.2\columnwidth]{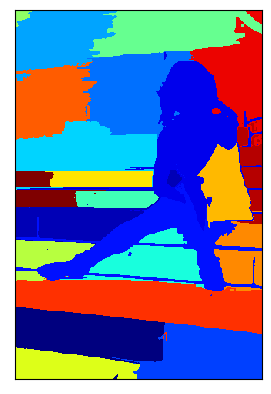}}
    \hspace{1em}
    \subfigure[\textsc{FAISS Exact}] {\includegraphics[width=0.2\textwidth]{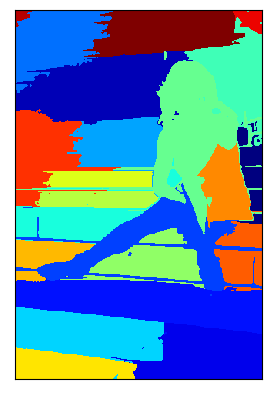}}
    \hspace{1em}
    \subfigure[\textsc{FAISS HNSW}] {\includegraphics[width=0.2\textwidth]{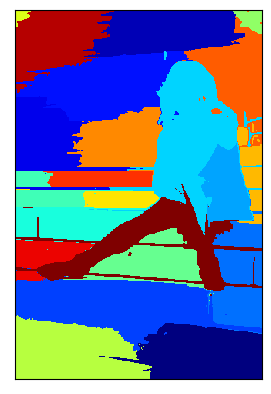}}
    \hspace{1em}
    \subfigure[\textsc{FAISS IVF}] {\includegraphics[width=0.2\textwidth]{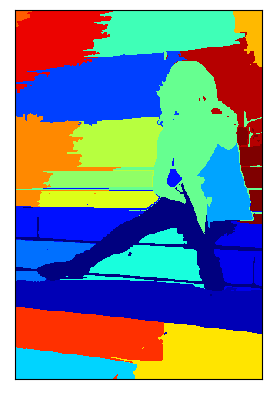}}

    \caption{More examples on the performance of the spectral clustering algorithms for image segmentation. \label{fig:bsdsapp2}}
\end{figure}

\end{document}